\documentclass[a4paper,USenglish,cleveref,nameinlink,autoref,thm-restate]{lipics-v2021}

\usepackage[inline]{enumitem}

\newtheorem{lemma}[theorem]{Lemma}

\usepackage{thmtools,thm-restate}
\usepackage{microtype}

\newcommand{\OGDF}{\textsc{ogdf}\xspace}
\newcommand{\DOMUS}{\textsc{domus}\xspace}

\author{Giordano Andreola}{Roma Tre University, Rome, Italy}{giordano.andreola@uniroma3.it}{https://orcid.org/0009-0003-7406-3514}{}

\author{Susanna Caroppo}{Roma Tre University, Rome, Italy}{susanna.caroppo@uniroma3.it}{https://orcid.org/0009-0001-4538-8198}{}

\author{Giuseppe {Di Battista}}{Roma Tre University, Rome, Italy}{giuseppe.dibattista@uniroma3.it}{https://orcid.org/0000-0003-4224-1550}{}

\author{Fabrizio Grosso}{CeDiPa - University of Perugia, Perugia, Italy}{fabrizio.grosso@unipg.it}{https://orcid.org/0000-0002-5766-4567}{}

\author{Maurizio Patrignani}{Roma Tre University, Rome, Italy}{maurizio.patrignani@uniroma3.it}{https://orcid.org/0000-0001-9806-7411}{}

\author{Allegra Strippoli}{Roma Tre University, Rome, Italy}{allegra.strippoli@uniroma3.it}{}{}

\authorrunning{Andreola et al.}

\Copyright{Giordano Andreola, Susanna Caroppo, Marco D'Elia, Giuseppe {Di Battista}, Fabrizio Grosso, Maurizio Patrignani, and Allegra Strippoli}

\pdfoutput=1 
\hideLIPIcs  

\nolinenumbers 

\usepackage{amsmath,amsthm,amsfonts}
\usepackage{mathtools}
\usepackage{hyperref}
\usepackage[linewidth=1pt]{mdframed}

\Crefname{property}{Property}{Properties}
\Crefname{observation}{Observation}{Observations}
\Crefname{theorem}{Theorem}{Theorems}
\Crefname{section}{Section}{Sections}
\Crefname{figure}{Figure}{Figures}

\usepackage{graphicx} 
\usepackage{xcolor,soul}

\usepackage{xspace}
\usepackage{todonotes}
\usepackage{complexity}

\definecolor{defblue}{rgb}{0.121,0.47,0.705}
\definecolor{darkred}{rgb}{.5,0.27,0.27}
\DeclareTextFontCommand{\emph}{\color{defblue}\em}

\definecolor{lipicsblue}{rgb}{0.08235294118,0.3098039216,0.537254902}

\definecolor{linkblue}{rgb}{0.098,0.098,0.4392}
\definecolor{ourgreen}{rgb}{0.509,0.745,0.235}
\definecolor{ourgreen2}{rgb}{0.409,0.600,0.200}

\definecolor{indianred}{rgb}{0.804,0.361,0.361}
\definecolor{indianred1}{rgb}{1,0.416,0.416}
\definecolor{indianred3}{rgb}{0.804,0.333,0.333}
\definecolor{orangered}{rgb}{1,0.271,0}
\definecolor{coral1}{rgb}{1,0.447,0.337}
\definecolor{rosybrown2}{rgb}{0.933,0.231,0.231}
\definecolor{aquamarine4}{rgb}{0.271,0.545,0.455}
\definecolor{chartreuse3}{rgb}{0.4,0.804,0}
\definecolor{mediumpurple3}{rgb}{0.537,0.408,0.804}
\definecolor{mediumvioletred}{rgb}{0.78,0.082, 0.522}

\hypersetup{colorlinks=true,
    linkcolor=orangered,
    anchorcolor=lipicsblue,
    citecolor=lipicsblue,
    filecolor=lipicsblue,
    menucolor=lipicsblue,
    urlcolor=lipicsblue,
    bookmarksopen=true,
    bookmarksopenlevel=2,
    bookmarksnumbered=true,
    plainpages=false,
    }

\Crefname{claim}{Claim}{Claims}
\Crefname{figure}{Fig.}{Figs.}
\Crefname{section}{Sec.}{Sec.}

\date{}

\funding{This research was supported, in part,  by MUR of Italy (PRIN Project no.~2022ME9Z78~-- NextGRAAL and PRIN Project no.~2022TS4Y3N~-- EXPAND). The fifth author was supported by Ce.Di.Pa. - PNC Programma unitario di interventi per le aree  del terremoto del 2009-2016 - Linea di intervento 1 sub-misura B4 - "Centri di ricerca per l'innovazione" CUP J37G22000140001.}

\title{A Walk on the Wild Side: a Shape-First Methodology for Orthogonal Drawings}


\EventEditors{Vida Dujmovic and Fabrizio Montecchiani}
\EventNoEds{2}
\EventLongTitle{33rd International Symposium on Graph Drawing and Network Visualization (GD 2025)}
\EventShortTitle{GD 2025}
\EventAcronym{GD}
\EventYear{2025}
\EventDate{September 24--26, 2025}
\EventLocation{Norrköping, Sweden}
\EventLogo{}
\SeriesVolume{100}
\ArticleNo{100}

\begin{document}

\maketitle


\begin{abstract}
Several algorithms for the construction of orthogonal drawings of graphs, including those based on the Topology-Shape-Metrics (TSM) paradigm, tend to prioritize the minimization of crossings. This emphasis has two notable side effects: 
some edges are drawn with unnecessarily long sequences of segments and bends, and the overall drawing area may become excessively large.
As a result, the produced drawings often lack geometric uniformity.
Moreover, orthogonal crossings are known to have a limited impact on readability, suggesting that crossing minimization may not always be the optimal goal.
In this paper, we introduce a methodology that ``subverts'' the traditional TSM pipeline by focusing on minimizing bends. Given a graph $G$, we ideally seek to construct a rectilinear drawing of $G$, that is, an orthogonal drawing with no bends. When not possible, we incrementally subdivide the edges of $G$ by introducing dummy vertices that will (possibly) correspond to bends in the final drawing. This process continues until a rectilinear drawing of a subdivision of the graph is found, after which the final coordinates are computed.
We tackle the (NP-complete) rectilinear drawability problem by encoding it as a SAT formula and solving it with state-of-the-art SAT solvers. If the SAT formula is unsatisfiable, we use the solver’s proof to determine which edge to subdivide.
Our implementation, \DOMUS, which is fairly simple, is evaluated through extensive experiments on small- to medium-sized graphs. The results show that it consistently outperforms \OGDF’s TSM-based approach across most standard graph drawing metrics.
\end{abstract}

\clearpage
\section{Introduction}\label{sec:introduction}


\emph{Orthogonal drawings}, where edges are represented as chains of horizontal and vertical segments, are a foundational topic in Graph Drawing, valued both for their practical applications and for the rich body of research they have inspired over the past four decades.

The most well-known and widely studied methodology for constructing orthogonal drawings is the \emph{Topology-Shape-Metrics} (\emph{TSM}) approach (see, e.g., \cite{DBLP:books/ph/BattistaETT99,DBLP:conf/er/TamassiaBT83,DBLP:journals/tsmc/TamassiaBB88}). TSM generates an orthogonal drawing in three steps.
%
First, a planar embedding (\emph{topology}) of the graph is computed, and \emph{dummy vertices} are possibly introduced to represent edge crossings.
%
Second, a \emph{shape}—minimizing the number of \emph{bends} (points where two edge segments meet at a right angle)—is computed, see~\cite{DBLP:journals/siamcomp/Tamassia87}.
Third, an orthogonal drawing (\emph{metrics}) consistent with the computed shape is constructed.
Although originally designed for graphs with maximum degree $4$, TSM has been extended to general graphs by either representing vertices as boxes sized proportionally to their degree (see, e.g., \cite{DBLP:journals/tse/BatiniNT86,DBLP:journals/jss/BatiniTT84}), or by allowing partial edge overlaps, which requires adapting the flow-based optimization accordingly (see, e.g., \cite{DBLP:phd/de/Eiglsperger2004,DBLP:conf/gd/FossmeierK95}).

While TSM is arguably the most well-known method,  alternative approaches have been proposed for orthogonal drawings (see, e.g., \cite{Biedl-Kant-conference-94,Biedl-Kant-98,DBLP:conf/gd/PapakostasT94,DBLP:journals/tc/PapakostasT98,tamassia-orthogonal-visibility}).
%
%
%
Nonetheless, the comprehensive and still influential quantitative aggregate~\cite{DBLP:journals/cgf/BartolomeoCSPWD24} evaluation in~\cite{DBLP:journals/comgeo/WelzlBGLTTV97} shows that the TSM approach outperforms all the above cited alternative methods across nearly all standard graph drawing metrics.
Note that other algorithms (see e.g.,~\cite{Kieffer-2016}) have been proposed and assessed not through quantitative measures but via user studies. Also, the algorithm in \cite{DBLP:conf/gd/HegemannW23} neglects the experiments in~\cite{DBLP:journals/comgeo/WelzlBGLTTV97} in favor of a comparison with~\cite{Kieffer-2016}.

Furthermore, the success of TSM is underscored by its implementation in widely used open-source libraries such as \OGDF~\cite{chimani2013open} and GDToolkit~\cite{DBLP:reference/crc/BattistaD13}, as well as its adoption in the Graph Drawing engine of yWorks~\cite{DBLP:phd/de/Eiglsperger2004,DBLP:conf/gd/FossmeierK95,DBLP:conf/ipco/KlauM99,DBLP:journals/networks/ResendeR97}, a leading company in graph visualization.



%

However, two key observations can be made.
(1) Existing algorithms, including those based on TSM, prioritize minimizing crossings. This often leads to unnecessarily long edge paths with many bends and results in drawings with large area and poor geometric uniformity.
(2) Several studies (e.g.,~\cite{Huang-Effects_of_Crossing_Angles}) have shown that orthogonal crossings, where edges meet at right angles, have minimal impact on readability. This suggests that aggressively minimizing crossings in orthogonal drawings may not always be the optimal design goal.

In this paper, we start from the two aforementioned observations and introduce a methodology, called \emph{Shape-Metrics} (\emph{SM}), that ``subverts'' the traditional TSM pipeline. 
SM focuses first on bends. Namely, given a graph $G$, we ideally seek to construct a \emph{rectilinear drawing} of $G$, that is, a drawing in which all edges are represented by straight horizontal or vertical segments with no bends. When such a drawing is not feasible, we incrementally subdivide the edges of $G$ by introducing dummy vertices. Each of these vertices will (possibly) correspond to a bend in the final drawing. We continue this process until we obtain a subdivision of $G$ that admits a rectilinear drawing. Once such a subdivision is found, the final metric phase computes the actual coordinates of the drawing.


SM is based on the concept of \emph{shaped} graph, which we define as an assignment of a label from the set ${\cal L} = \{ L, R, D, U\}$ to each of its edges, indicating the direction—Left, Right, Down, or Up—of that edge in a drawing.
Given a shaped graph, it can be decided in polynomial time~\cite{DBLP:conf/gd/ManuchPPT10} if it admits a rectilinear drawing with that shape.
On the other hand, checking if a graph admits a rectilinear drawable shape is NP-complete~\cite{DBLP:conf/gd/EadesHP09}.
%
We exploit a new simple characterization of rectilinear drawable shaped graphs. Informally, a shaped cycle is said to be \emph{complete} if its edges have all the four labels; we show that a shaped graph is rectilinear drawable if and only if all its cycles are complete.

We tackle the NP-hardness barrier by formulating the rectilinear drawability testing problem as a Boolean satisfiability (SAT) problem and leveraging the power of modern SAT solvers.
SAT solvers have already proven to be effective in combinatorics and geometry (see, e.g.,  \cite{DBLP:journals/corr/abs-2504-04821,DBLP:journals/combinatorics/DaubelJMS19,DBLP:journals/jgaa/MutzeS20,DBLP:journals/comgeo/Scheucher20,DBLP:journals/corr/abs-1811-06482}). However, as far as we know, SAT solvers have rarely been applied to Graph Drawing. The related works we are aware of are \cite{DBLP:conf/gd/BiedlBNNPR13}, which employs a \emph{hybrid} ILP/SAT approach to compute $st$-orientations and visibility representations, \cite{DBLP:conf/gd/Bekos0Z15} where a SAT solver is used to compute book embeddings, and \cite{DBLP:conf/gd/ChimaniZ12,
DBLP:conf/gd/ChimaniZ13,DBLP:journals/jea/ChimaniZ15} where a SAT solver is used to test upward planarity.


There are differences and analogies between SM and TSM. TSM leverages the fact that checking if a graph is planar can be done efficiently. On our side we have that, unfortunately, checking if a graph is rectilinear drawable is NP-complete~\cite{DBLP:conf/gd/EadesHP09}. On the other hand both methodologies use specific properties of shapes: TSM's cornerstone is the balance of turns in shapes of faces of planar graphs \cite{doi:10.1137/0214027}, while SM's cornerstone is the completeness of cycles. Also, both methodologies augment the graph to be drawn with dummy vertices, TSM to represent crossings, SM to represent bends.

We implemented SM in a tool called \DOMUS (\emph{Drawing Orthogonal Metrics Using the Shape}) and experimented its effectiveness against the TSM implementation available in \OGDF~\cite{chimani2013open} measuring widely adopted Graph Drawing metrics \cite{DBLP:journals/cgf/BartolomeoCSPWD24,DBLP:journals/comgeo/WelzlBGLTTV97}. Although other comparable tools exist (e.g., GDToolkit~\cite{DBLP:reference/crc/BattistaD13}), we selected \OGDF as our benchmark due to its open-source nature, widespread adoption within the Graph Drawing community, and the fact that it has benefited from over a decade of active development and optimization.
In view of the results in~\cite{DBLP:journals/comgeo/WelzlBGLTTV97} it was pointless to compare \DOMUS with the algorithms in \cite{Biedl-Kant-conference-94,Biedl-Kant-98,DBLP:conf/gd/PapakostasT94,DBLP:journals/tc/PapakostasT98,tamassia-orthogonal-visibility}. 
Also, the algorithms in \cite{DBLP:conf/gd/HegemannW23,Kieffer-2016} adopt a definition of orthogonal drawing which is different from the widely adopted one given in this paper. E.g., even if a vertex has degree less or equal than $4$ more than one of its incident edges can exit from the same direction. Hence, comparing \DOMUS with them would be unfair.

We did two sets of experiments, which we refer to as ``in-vitro'' and ``in-the-wild''.
The in-vitro experiments aim to thoroughly evaluate the performance and characteristics of SM, including the use of the SAT solver, on a dataset of random graphs with predefined sizes and densities. These graphs have a maximum degree of $4$, which is the baseline for orthogonal drawings, and contain $20-60$ vertices, a range considered important for node-link representations in graph visualization (see, e.g.,~\cite{DBLP:journals/ivs/GhoniemFC05}).
%
The in-the-wild experiments use the Rome Graphs dataset \cite{DBLP:journals/comgeo/WelzlBGLTTV97}, which is one of the most popular collections of
graphs~\cite{DBLP:journals/cgf/BartolomeoCSPWD24}, comprised of $11,534$ graphs, having $10-100$ vertices and $9-158$ edges, with maximum degree $13$.
The experiments show that for most of the typical graph drawing metrics~\cite{DBLP:journals/cgf/BartolomeoCSPWD24}, \DOMUS outperforms the TSM algorithm implemented in \OGDF.
Moreover, SM is simple to implement: it reduces shape constraints to a SAT formula, invokes an off-the-shelf SAT solver, and computes the final drawing metrics using a slight variation of standard TSM compaction techniques.
\cref{fig:two-drawings-of-the-same-graph} shows two drawings of the same graph, one produced with \OGDF and the other with \DOMUS.

\begin{figure}[tb!]
    \centering
    \begin{subfigure}[b]{0.099\textwidth}
    \centering
        \includegraphics[width=\textwidth]{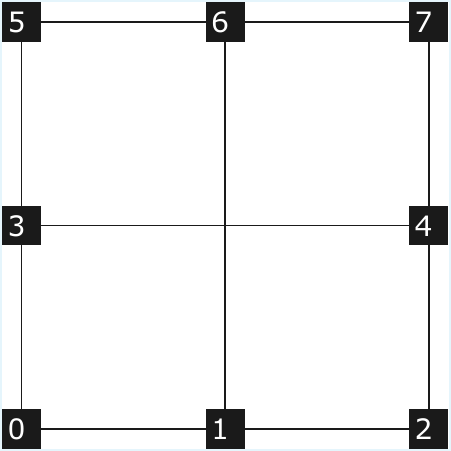}
        \subcaption{}
        \label{fig:rome_scatter_bends_comparison}
    \end{subfigure}
    \hfil
    \begin{subfigure}[b]{0.15\textwidth}
    \centering
        \includegraphics[width=\textwidth]{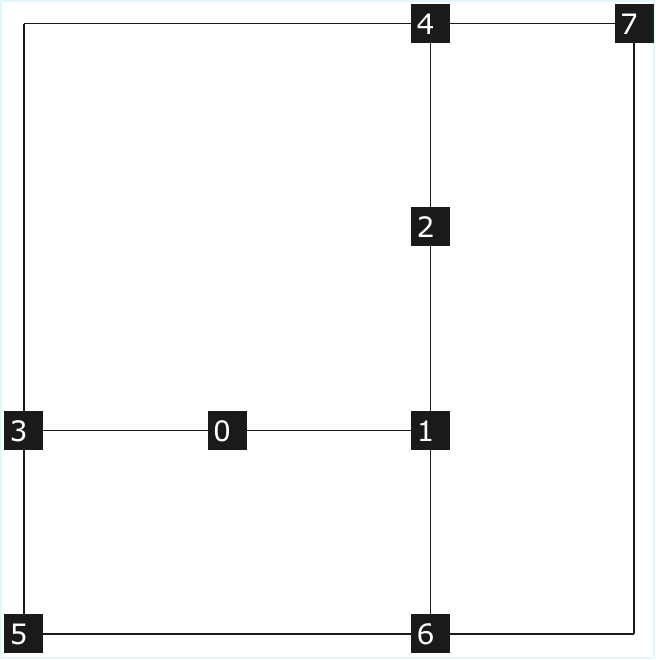}
        \subcaption{}
        \label{fig:rome_scatter_crossings_comparison}
    \end{subfigure}
    \caption{The graph in Fig.\ 2.4 of \cite{DBLP:books/ph/BattistaETT99}, represented: {\bf(a)}  with \DOMUS and {\bf(b)} with \OGDF.}
    \label{fig:two-drawings-of-the-same-graph}
\end{figure}

The paper is organized as follows.
Preliminaries are in \cref{sec:preliminaries}.
\cref{sec:shape-rectilinear-drawings} illustrates the characterization of rectilinear drawable shaped graphs which is the milestone of SM. 
\cref{sec:methodology} describes SM, clarifying how its steps are activated and implemented and the role played by the SAT-solver.
\cref{sec:experiments-4} and \cref{sec:different-settings} show the in-vitro and the in-the-wild experiments, respectively.
Conclusions and open problems are proposed in \cref{sec:conclusions}.
Because of space limitations some proofs are only sketched. A full version of such proofs and other extra material is in the appendix.

\section{Preliminaries}\label{sec:preliminaries}

For basic graph drawing terminology and definitions refer, e.g., to~\cite{DBLP:books/ph/BattistaETT99,DBLP:books/ws/NishizekiR04}.

Let $G$ be a graph.
with vertex set $V(G)$ and edge set $E(G)$.
%
For technical reasons, we assign an arbitrary (e.g., random) orientation to each edge in $E(G)$; thus, if an edge is denoted as $(u,v)$, it is considered as directed from $u$ to $v$.
However, when referring to paths, cycles, vertex degrees, or connectivity in $G$, we disregard these orientations.

A \emph{cut-vertex} is a vertex whose removal disconnects $G$. Graph $G$ is \emph{biconnected} if it has no cut-vertex. A \emph{biconnected component} of $G$ is a maximal (in terms of vertices and edges) biconnected subgraph of~$G$. A biconnected component is \emph{trivial} if it consists of a single edge and \emph{non-trivial} otherwise.
Unless otherwise specified we assume that graphs are connected.

In a \emph{rectilinear drawing} of $G$:
\begin{enumerate*}[label=\bf(\roman*)]
    \item Each vertex $v \in V(G)$ is represented by a point with integer coordinates, denoted by $x(v)$ and $y(v)$.
    \item No two vertices are represented by the same point.
    \item Each edge is represented by either a horizontal or a vertical line segment connecting the points that represent its end-vertices.
    \item The interior of a segment representing an edge does not intersect any point representing a vertex.
\end{enumerate*}
A graph is said to be \emph{rectilinear drawable} if it admits a rectilinear drawing. Note that such a drawing requires each vertex to have degree at most $4$. Hence, unless otherwise specified, we will consider graphs with maximum vertex degree $4$.
%

We use the label set ${\cal L} = \{ L, R, D, U\}$ to denote directions, where $L$, $R$, $D$, and $U$ stand for left, right, down, and up, respectively. Also, we say that $L$ and $R$ are \emph{opposite} labels, and likewise, $D$ and $U$ are \emph{opposite} labels. For any label $A\in \cal L$, we denote its opposite by $\overline{A}$.

A \emph{shape} of a graph $G$ is a function $\lambda$ that assigns to each edge $(u,v) \in E(G)$ a label from the set~${\cal L}$.
Consider an edge $(u,v)$. Labeling $(u,v)$ with, say, $R$ indicates that in the shape of $G$, the edge is directed to the right when traversed from $u$ to $v$. Clearly, if the edge is traversed in the opposite direction (from $v$ to $u$) the direction should be left.
Hence, with a little abuse of notation, if $\lambda(u,v)=R$ (resp. $\lambda(u,v)=L$) we say that $\lambda(v,u)=L$ (resp. $\lambda(v,u)=R$). The other labels are treated analogously.
%
%
Additionally, for each pair of edges that share a vertex $v$ it is mandatory that $\lambda(v,u) \neq \lambda(v,w)$.
This constraint ensures that two edges sharing the same vertex ``cannot point in the same direction'', preventing the two edges from ``overlapping''. 
%
We define a \emph{shaped graph} to be a graph with a shape.


    







A shape can be assigned to a graph by starting from one of its rectilinear drawings, if such a drawing exists.
Namely, given a rectilinear drawing $\Gamma$ of a graph $G$ we can label the edges $E(G)$ as follows.
For each $(u,v) \in E(G)$:
\begin{enumerate*}[label=\bf(\roman*)]
    \item If $x(u)<x(v)$ we set $\lambda(u,v)=R$, 
    \item If $x(u)>x(v)$ we set $\lambda(u,v)=L$,
    \item If $y(u)<y(v)$ we set $\lambda(u,v)=U$,
    \item If $y(u)>y(v)$ we set $\lambda(u,v)=D$.
\end{enumerate*}
We say that this is the \emph{shape} of $\Gamma$. It is easy to check that this labeling is consistent with the definition of shape, and hence it is also a valid shape for $G$.

A shaped graph is \emph{rectilinear drawable} if it admits a rectilinear drawing having exactly its shape. The shaped graph in \cref{fig:Original-graph} is not rectilinear drawable. \cref{fig:not-rectilinear-drawable} shows a failed attempt to construct a rectilinear drawing of such a shaped graph.

\begin{figure}[tb!]
    \centering
    \begin{subfigure}[t]{0.3\textwidth}
        \includegraphics[width=\textwidth, page=8]{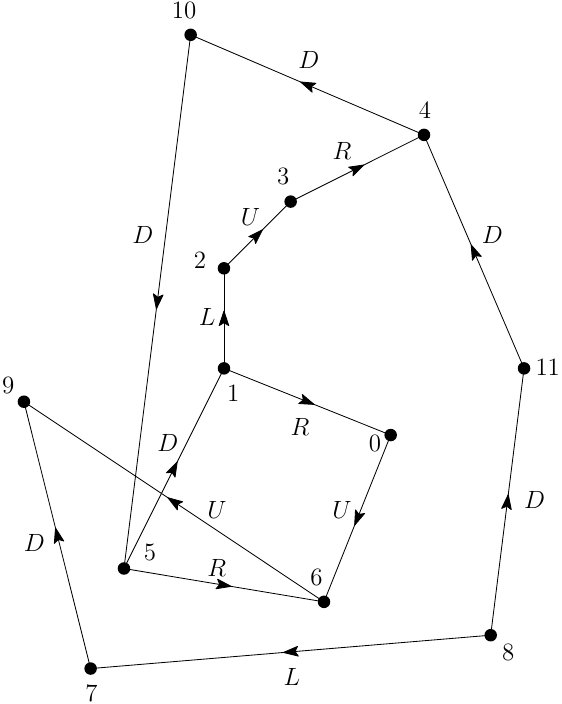}
        \subcaption{}
        \label{fig:Original-graph}
    \end{subfigure}
    \quad
    \begin{subfigure}[t]{0.3\textwidth}
        \includegraphics[width=\textwidth, page=7]{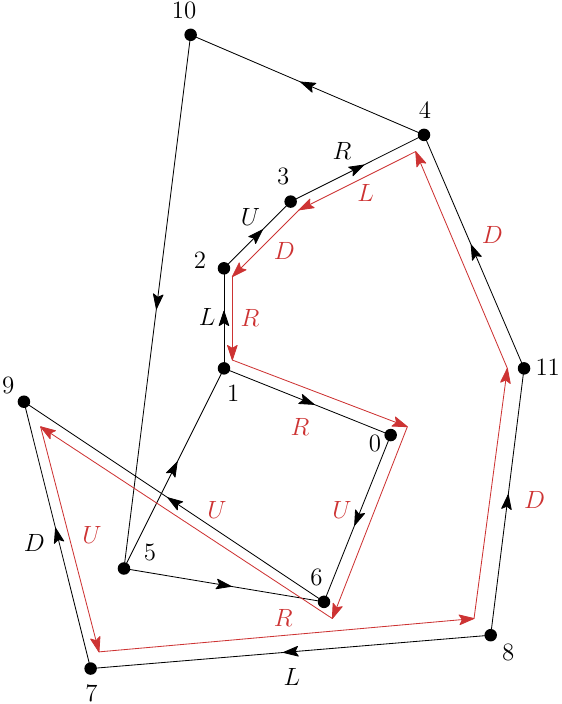}
        \subcaption{}
        \label{fig:cycle-red}
    \end{subfigure}
    \quad
    \begin{subfigure}[t]{0.3\textwidth}
        \includegraphics[width=\textwidth, page=9]{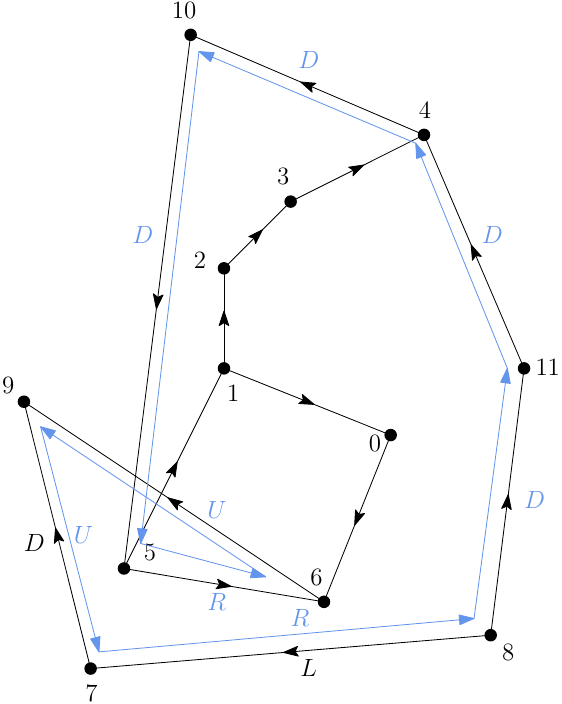}
        \subcaption{}
        \label{fig:cycle-blue}
    \end{subfigure}
    \medskip
    \begin{subfigure}[t]{0.3\textwidth}
        \includegraphics[width=\textwidth, page=10]{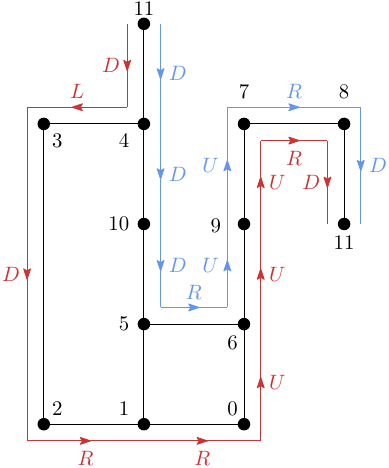}
        \subcaption{}
        \label{fig:not-rectilinear-drawable}
    \end{subfigure}
    \quad
    \begin{subfigure}[t]{0.3\textwidth}
        \includegraphics[width=\textwidth, page=11]{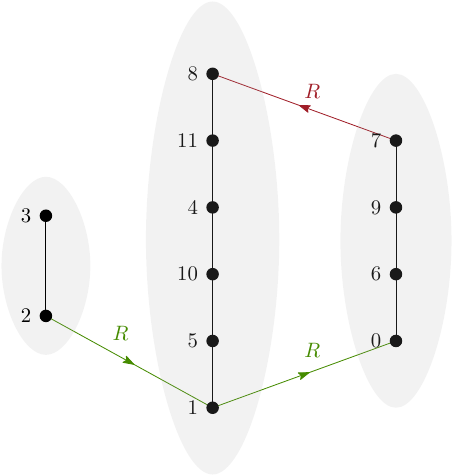}
        \subcaption{}
        \label{fig:cycle-auxiliary}
    \end{subfigure}
    \caption{{\bf(a)} A shaped graph $G$. {\bf(b)} A complete cycle of $G$. {\bf(c)} Another cycle of $G$. {\bf(d)} A failed attempt to draw $G$ according to its shape. {\bf(e)} Graph $G_x$ for the shaped graph of \cref{fig:Original-graph}.}
    \label{fig:not-drawable}
\end{figure}

Given a graph $G$ a \emph{subdivision} of $G$ is a graph $G'$ obtained by replacing certain edges of $E(G)$ with internally vertex-disjoint simple paths. The new vertices in $V(G')\setminus V(G)$ introduced along these paths are called \emph{dummy vertices}.
An \emph{orthogonal drawing} $\Gamma$ of $G$ is a rectilinear drawing of a subdivision of $G$. Let $v$ be a dummy vertex, if the two segments incident to $v$ in $\Gamma$ are one horizontal and one vertical, then $v$ is called a \emph{bend}.

\section{Rectilinear Drawable Shapes and the Shape of Cycles }\label{sec:shape-rectilinear-drawings}

The methodology we present in \cref{sec:methodology} is based on the results presented in this section.

Consider a shaped simple cycle $c = v_0, v_1, \dots, v_{p-1}$, with $p \geq 4$, of $G$. We say that $c$ is \emph{complete} if it contains four pairs of vertices $v_i,v_{i+1}$, $v_j,v_{j+1}$, $v_h,v_{h+1}$, and $v_k,v_{k+1}$ such that:
\begin{enumerate*}[label=\bf(\roman*)]
    \item $\lambda(v_i,v_{i+1})=L$,
    \item $\lambda(v_j,v_{j+1})=R$,
    \item $\lambda(v_h,v_{h+1})=D$,
    \item $\lambda(v_k,v_{k+1})=U$,
\end{enumerate*}
where the vertices indexes of $c$ are intended $\text{mod~} p$. An example of a complete cycle is shown in \cref{fig:cycle-red}, while \cref{fig:cycle-blue} illustrates a cycle that is not complete.



\begin{restatable}{theorem}{thCycle}\label{th:cycle}

A shaped simple cycle $c=v_0, v_1, \dots, v_{p-1}$ with $p \geq 4$ is rectilinear drawable if and only if it is complete.

\end{restatable}

\begin{proof}[Proof sketch]
If $c$ admits a rectilinear drawing, then intersecting it with horizontal and vertical lines (avoiding vertices) reveals edges labeled $U$, $D$, $L$, and $R$. Hence, $c$ is complete.
Conversely, suppose $c$ is complete. We prove that it admits a rectilinear drawing by induction on the number of vertices $p$.
For $p = 4$, the completeness implies the presence of all directions in $\{L, R, U, D\}$. The only valid labelings without consecutive opposites are $R$–$U$–$L$–$D$ and $R$–$D$–$L$–$U$, which form rectangles.
Now, assume that any complete cycle with $p - 2$ vertices ($p \geq 6$) is drawable. Let $c$ be a complete cycle with $p\geq6$ vertices. Without loss of generality, assume no two consecutive labels are equal. Then, the label sequence $\Lambda = s_0, \dots, s_{p-1}$ alternates between $\{L, R\}$ and $\{U, D\}$, contains all four labels, and avoids consecutive equal or opposite pairs.
Then, there exists a removable pair $s_i, s_{i+1}$ such that the resulting cycle $c'$ remains complete. By induction, $c'$ is drawable. Reinserting $s_i, s_{i+1}$ as a right-angle bend (e.g., $\langle U, R \rangle$, $\langle D, L \rangle$) and applying simple geometric adjustments yields a drawing of $c$ (see~\cref{fig:insertion}).
Thus, $c$ is rectilinear drawable. \end{proof}

\begin{figure}[tb!]
    \centering
    \begin{subfigure}[b]{0.19\textwidth}
    \centering
        \includegraphics[width=\textwidth, page=1]{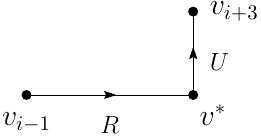}
        \subcaption{}
        \label{fig:Gamma'}
    \end{subfigure}
    \hfill
    \begin{subfigure}[b]{0.19\textwidth}
    \centering
        \includegraphics[width=\textwidth, page=2]{img-figures.pdf}
        \subcaption{}
        \label{fig:insertion_D-L}
    \end{subfigure}
    \hfill
    \begin{subfigure}[b]{0.19\textwidth}
    \centering
        \includegraphics[width=\textwidth, page=3]{img-figures.pdf}
        \subcaption{}
        \label{fig:insertion_U-R}
    \end{subfigure}
        \hfill
    \begin{subfigure}[b]{0.19\textwidth}
    \centering
        \includegraphics[width=\textwidth, page=4]{img-figures.pdf}
        \subcaption{}
        \label{fig:insertion_D-R}
    \end{subfigure}
        \hfill
    \begin{subfigure}[b]{0.19\textwidth}
    \centering
        \includegraphics[width=\textwidth, page=5]{img-figures.pdf}
        \subcaption{}
        \label{fig:insertion_U-L}
    \end{subfigure}
    \caption{Illustration for the proof of \cref{th:cycle}. {\bf(a)} The partial rectilinear drawing of $c'$ according to the new sequence. The result of the extension of the drawing of $c'$ to $c$ by re-inserting labels {\bf(b)} $D$ and $L$; {\bf(c)} $U$ and $R$; {\bf(d)} $D$ and $R$; and {\bf(e)} $U$ and $L$.
    }
    \label{fig:insertion}
\end{figure}



To test if a shaped graph is rectilinear drawable we can use the following theorem, which has been proved in \cite{DBLP:conf/gd/ManuchPPT10}. 

\begin{theorem}\label{th:manuch} \cite[Theorem 3]{DBLP:conf/gd/ManuchPPT10}.
Given a shaped graph $G$, there exists an $O(|V(G)|\cdot|E(G)|)$ algorithm to test if it is rectilinear drawable. In the positive case a rectilinear drawing can be found in $O(|V(G)|\cdot|E(G)|)$ time.
\end{theorem}

The proof in \cite{DBLP:conf/gd/ManuchPPT10} is based on constructing, starting from the given shaped graph, two auxiliary directed graphs which we call $G_x$ and $G_y$. For the sake of readability, we call the vertices and the edges of such graphs \emph{nodes} and \emph{arcs}, respectively.

We say that two vertices $v_0$ and $v_{k-1}$ of $G$ are \emph{$x$-aligned} if a path $v_0, \dots, v_{k-1}$ exists in $G$ such that $\lambda(v_i,v_{i+1})=D$ ($i=0,\dots,k-2$). Also, we say that a vertex of $G$ is \emph{$x$-aligned} with itself.
Graph $G_x$ is defined as follows. A node of $G_x$ is a maximal set of $x$-aligned vertices of $G$. Given two nodes $\mu$ and $\nu$ of $G_x$, by the definition of shape, they are disjoint sets. There is an arc from node $\mu$ to node $\nu$ if there are two vertices $u$ and $v$ of $G$ such that $u \in \mu$, $v \in \nu$, and either $(u,v) \in E$, and $\lambda(u,v)=R$ or $(v,u) \in E$, and $\lambda(v,u)=L$. \cref{fig:cycle-auxiliary} shows graph $G_x$ for the shaped graph in \cref{fig:Original-graph}.
Graph $G_y$ is defined analogously.

%

In \cite{DBLP:conf/gd/ManuchPPT10} it is proved that if there is an arc in $G_x$ (resp., $G_y$) from $\mu$ to $\nu$, then in any rectilinear drawing of $G$ the vertices in $\mu$ should be placed to the left of (resp., above) the vertices in $\nu$. Also all the vertices in the same node should be vertically (resp., horizontally) aligned. This implies that a necessary condition for the shape to be rectilinear drawable is that both $G_x$ and $G_y$ are acyclic. Observe that graph $G_x$ of \cref{fig:cycle-auxiliary} is cyclic and the shaped graph in \cref{fig:Original-graph} is not rectilinear drawable.
The authors show that this condition is also sufficient by constructing a drawing. Namely, they assign the $x$-coordinates to the vertices of $V(G)$ as follows: (1) the same $x$ coordinate is assigned to all the vertices corresponding to the same node of $G_x$, (2) different coordinates are assigned to vertices corresponding to different nodes of $G_x$, and (3) $x$-coordinates  increase according to the orientations of the nodes of $G_x$. The $y$-coordinates are assigned analogously. Since in our case the vertices of $V(G)$ have degree at most $4$, it easy to see that the construction of $G_x$ and $G_y$, the acyclicity test, and, in case, the construction of the drawing can be accomplished in $O(|V(G)|)$ time.

We are now able to prove the following characterization.

\begin{restatable}{theorem}{thDrawable}\label{th:drawable}
A shaped graph $G$ is rectilinear drawable if and only if all its simple cycles are complete. In case $G$ is not rectilinear drawable, a simple cycle of $G$ which is not complete can be found in $O(|V(G)|)$ time.
\end{restatable}

\begin{proof}[Proof sketch]
    The necessity follows from \cref{th:cycle}.
    To prove the sufficiency consider a shaped graph $G$ and its auxiliary directed graphs $G_x$ and $G_y$. From the proof of \cref{th:manuch} we have that the presence of a cycle in $G_x$ or $G_y$ implies that $G$ is not rectilinear drawable using the prescribed shape. We prove that a cycle in $G_x$ or in $G_y$ implies a non-complete simple cycle in $G$.
    Let $C_x = \mu_0,\dots,\mu_{p-1}$ be a cycle in $G_x$, possibly p=1.
    See \cref{fig:cycle-auxiliary}, where $p=2$.
    Consider the $p$ pairs $u_i,v_{i+1}$ of vertices of $G$, such that $u_i \in \mu_i $ and $v_{i+1} \in \mu_{i+1}$, $(u_i,v_{i+1})\in E(G)$, and $\lambda(u_i,v_{i+1})=R$, for $i\in\{0,\dots,p-1\}$ (where indices are taken $\text{mod~} p$).
    By the construction of $G_x$, there exists a path $\Pi_i$ in $G$ from $v_i$ to $u_i$ that contains edges all labeled $D$ or all labeled $U$. The concatenation of the paths and edges $\Pi_0,(u_0,v_1),\Pi_1,\dots,(u_{p-2},v_{p-1}),\Pi_{p-1},(u_{p-1},v_0)$ forms a simple cycle $c$ in $G$ that is not complete. The cycle of \cref{fig:cycle-auxiliary} contains vertices $8$, $11$, $4$, $10$, $5$, $1$, $0$, $6$, $9$, and $7$.
    The construction of $G_x$ and of $G_y$, the acyclicity test, the search of a cycle $C_x$, and the computation of $c$ can all be done in $O(|V(G)|)$ time. \end{proof}


\section{A Methodology for Constructing Orthogonal Drawings of Graphs}\label{sec:methodology}

\begin{figure}[tb!]
\centering
    \includegraphics[width=0.8\textwidth, page=13]{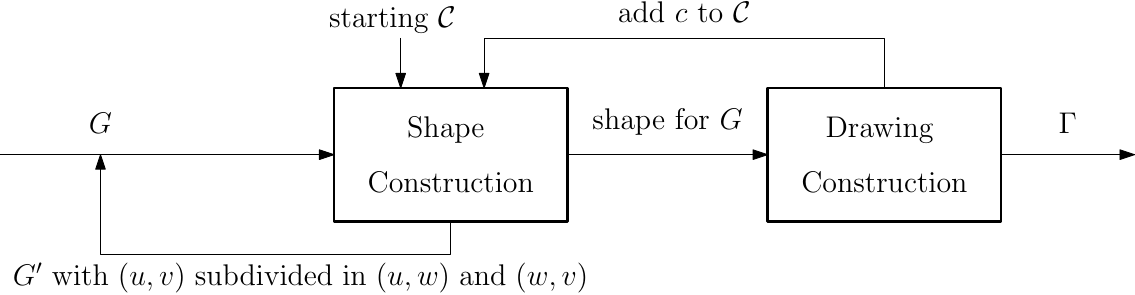}
    \caption{SM: a methodology for constructing orthogonal drawings of graphs.}
    \label{fig:construction-framework}
\end{figure}


This section presents SM. Given an $n$-vertex graph $G$ as input, it produces an orthogonal drawing $\Gamma$ of $G$ as output.
SM consists of two steps, called \emph{Shape Construction} and \emph{Drawing Construction}, see \cref{fig:construction-framework}. Each of these steps can be repeated multiple times.



\subsection{Shape Construction}

We know from \cref{th:drawable} that $G$ admits a rectilinear drawable shape if and only if it admits a shape in which all simple cycles are complete.
Consequently, the Shape Construction step may attempt to compute a shape for $G$ that enforces completeness for all simple cycles. However, since the number of simple cycles in $G$ can be exponential in the size of $G$, checking all of them is computationally infeasible.
To address this, the first phase of Shape Construction, called \emph{Cycle Selection}, selects a subset $\cal C$ of simple cycles in $G$ so that the Shape Construction searches for a shape that ensures completeness only for the cycles in $\cal C$, with the hope that this will be enough to enforce completeness for all cycles.

If no such shape exists, this implies that a rectilinear drawable shape for $G$ does not exist--since adding more constraints cannot make a previously unsolvable problem solvable. In this case, the Shape Construction step attempts to identify an edge $(u,v)$ that is ``responsible'' for the infeasibility, and splits it into two edges, $(u,w)$ and $(w,v)$, by introducing a dummy vertex $w$. It then restarts the process on the modified graph $G'$, which is a subdivision of $G$. The vertex $w$ will possibly be drawn as a bend in the final orthogonal drawing.
Additionally, splitting the edge $(u,v)$ indirectly alters the cycle set $\cal C$, producing a new set $\cal C'$, as the new vertex $w$ is now included in all cycles of $\cal C$ that originally contained $(u,v)$.

If such a shape exists, this does not necessarily imply that it is rectilinear drawable. Thus, the shape is passed as input to the Drawing Construction step for further verifications.

\subsection{Drawing Construction}

The Drawing Construction step relies on two theorems. \cref{th:manuch}, states that a shape can be efficiently tested for rectilinear drawability, and if the test is positive, a rectilinear drawing can be constructed efficiently. \cref{th:drawable}, states that if a shape is not rectilinear drawable a cycle $c$ that is not complete can be efficiently identified.
Hence, the Drawing Construction step either produces $\Gamma$ directly or suggests to the Shape Construction step to add $c$ to the cycle set $\cal C$, selectively enlarging the set of cycles considered for completeness.


Observe that, in general, due to the possible introduction of dummy vertices during the Shape Construction step, the resulting rectilinear drawing $\Gamma$ is the rectilinear drawing of a subdivision of $G$. Hence, $\Gamma$ is an orthogonal drawing of $G$.

\subsection{Implementing the Shape and the Drawing Construction}

Implementing the Shape Construction step essentially involves selecting an initial cycle set~$\cal C$ and checking whether a shape of $G$ exists that enforces completeness for all cycles in $\cal C$.

The selection of $\cal C$ can be done in various ways. A natural and intuitive strategy is to include cycles that cover all edges belonging to non-trivial biconnected components of $G$. Further details are given in \cref{sec:experiments-4}.

The checking for the existence of a shape is the most important part of the Shape Construction. We do this by converting completeness constraints into a CNF Boolean formula, $\mathcal{F}_{G, \mathcal{C}}$ and then by using a SAT solver to check satisfiability. This method benefits from (1) the efficiency of modern SAT solvers and (2) the solver’s ability to provide a proof if unsatisfiable. From this proof, we identify an edge label variable causing unsatisfiability, indicating the edge $(u,v)$ is over constrained, so we split it by adding a dummy vertex.

Formula ${\cal F}_{G, \cal C}$ has four boolean variables for each edge $(v,w)$ of $E(G)$. Such variables state if $\lambda(v,w)$ is equal to $L$, $R$, $D$, or $U$. It also has three types of clauses to state that: (1) each edge is assigned to exactly one label, (2) no vertex has two adjacent edges with the same direction, and (3) each simple cycle in $\mathcal{C}$ is complete. The full description of ${\cal F}_{G, \cal C}$ is given in \cref{sec:sat-formula-definition}.
It is easy to see that ${\cal F}_{G, \cal C}$ is satisfiable if and only if $G$ has a shape such that all the cycles in $\cal C$ are complete.

Implementing the Drawing Construction step primarily involves building the two auxiliary graphs $G_x$ and $G_y$ and checking each for acyclicity. If either graph contains a cycle, it can be used to identify a corresponding cycle $c$ in $G$ that is not complete.

Finally, we observe that the process described in the methodology is guaranteed to terminate, since an orthogonal drawing of a graph always exists (see, e.g, \cite{Biedl-Kant-conference-94,Biedl-Kant-98}).



\section{An Experimental Evaluation on Maximum Degree $4$ Graphs}\label{sec:experiments-4}

We implemented SM in our tool \DOMUS\ and evaluate its effectiveness through comprehensive ``in vitro'' experiments conducted using \DOMUS. Specifically, we describe:
\begin{enumerate*}[label={\bf (\roman*)}]
\item\label{li-implememtation-choices} the implementation choices in \DOMUS\ that define some aspects of SM,
\item\label{li-adversary} the benchmark adversary for comparison,
\item\label{li-set-of-graphs} the graph dataset and its key characteristics,
\item\label{li-computational-platform} the computational platform used,
\item\label{li-metrics} the evaluation metrics,
\item\label{li-experimental-results} the experimental results for each metric, and
\item\label{li-internals} internal performance indicators, such as SAT solver invocations per graph.
\end{enumerate*}

Two drawings of a graph in the dataset are in \cref{fig:domus-ogdf}.

\begin{figure}[tb!]
    \centering
    \begin{subfigure}[b]{0.264\textwidth}
    \centering
        \includegraphics[width=\textwidth]{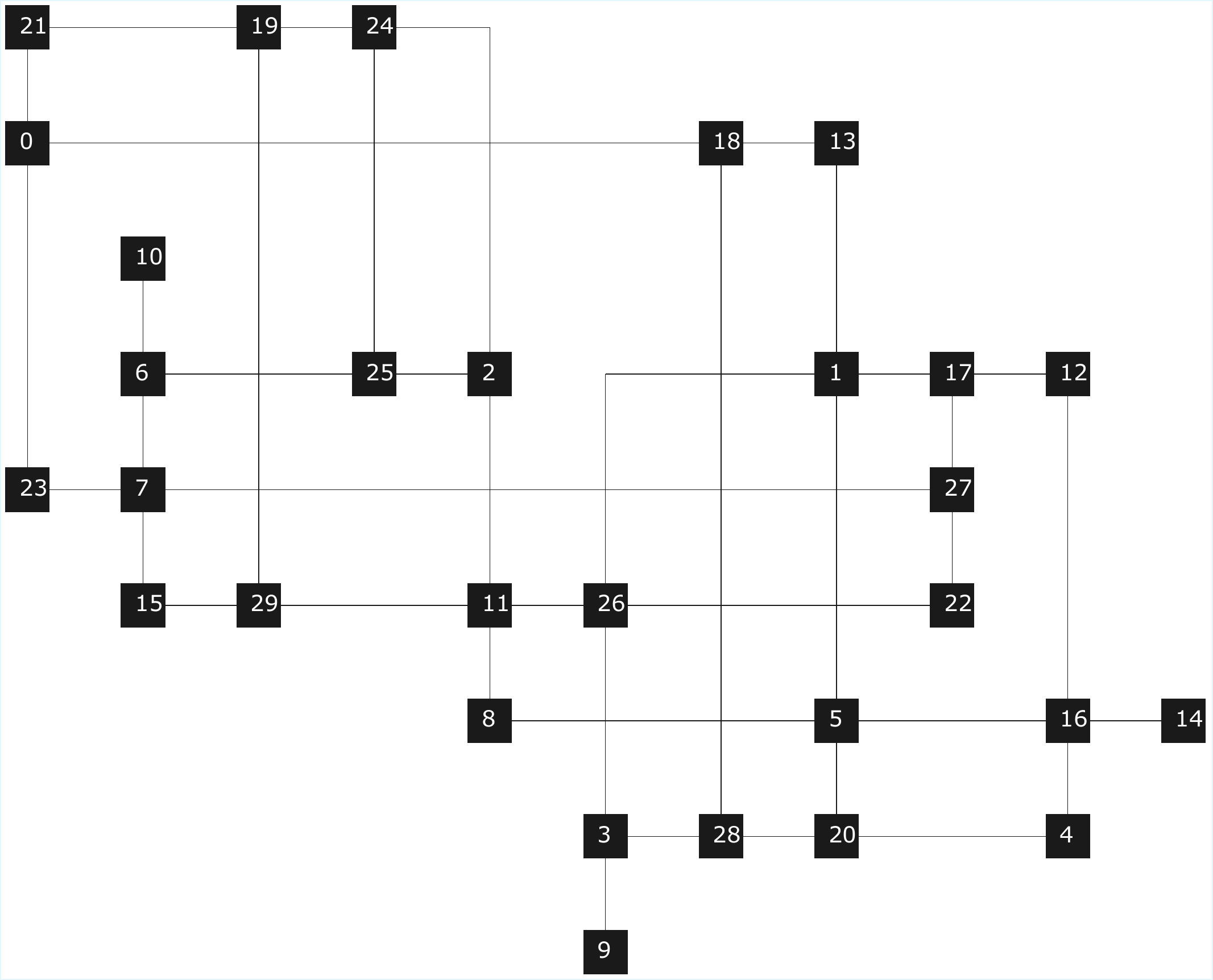}
        \subcaption{}
        \label{fig:domus}
    \end{subfigure}
    \hfil
    \begin{subfigure}[b]{0.363\textwidth}
        \centering
        \includegraphics[width=\textwidth]{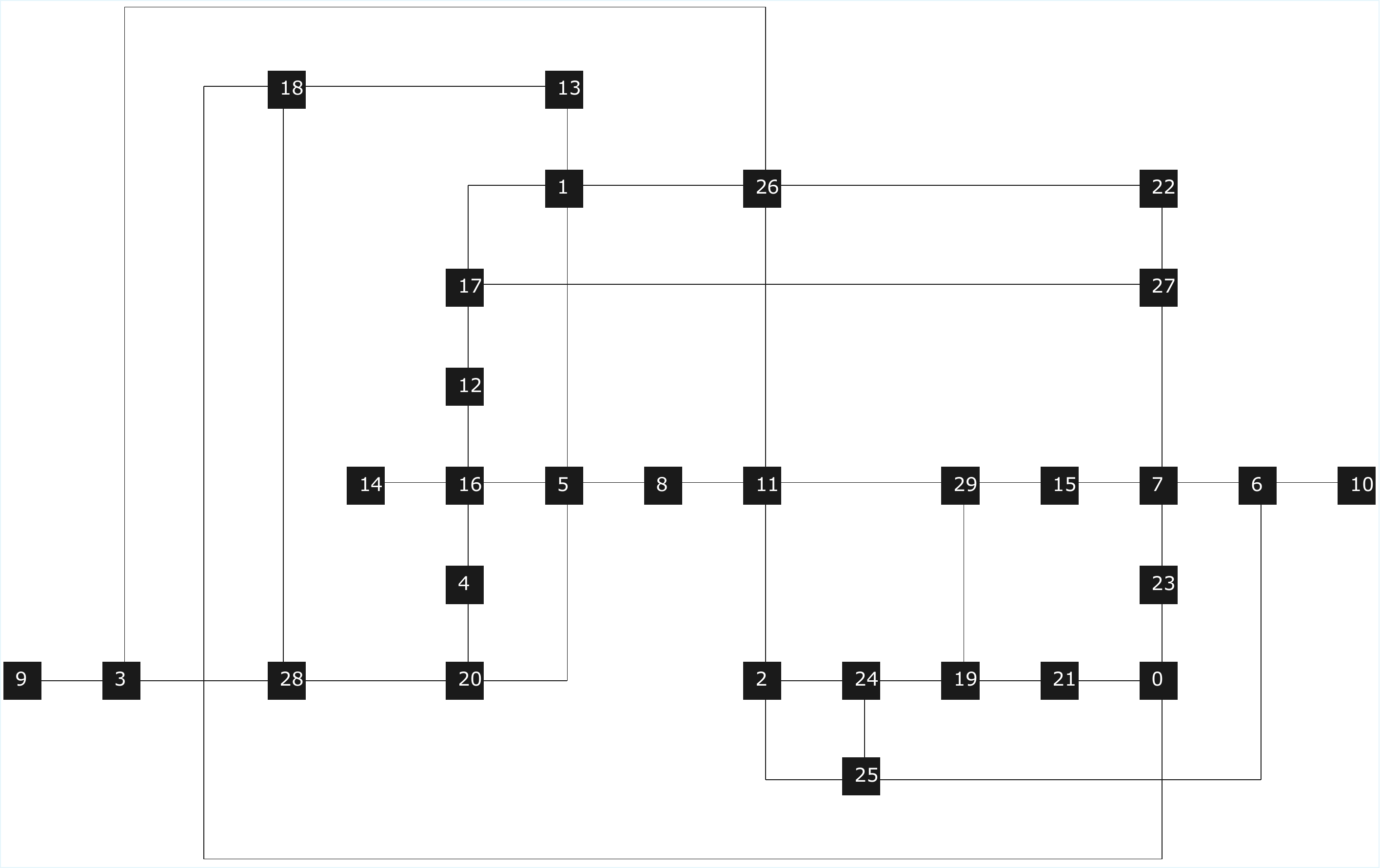}
        \subcaption{}
        \label{fig:ogdf}
    \end{subfigure}
    \caption{A graph in the dataset: {\bf(a)} drawn by \DOMUS and {\bf(b)} drawn by \OGDF.}
    \label{fig:domus-ogdf}
\end{figure}

\noindent\ref{li-implememtation-choices} In \DOMUS, we made specific implementation choices regarding the SAT solver, the initial cycle set $\cal C$, coordinate computation, and randomization.
For the SAT solver, we chose Glucose \cite{DBLP:journals/ijait/AudemardS18} because it is open source and it is well known for its efficiency. Importantly, Glucose provides proofs of unsatisfiability, which we use to identify edges to split, thereby increasing flexibility in the drawing.
%
%
To initialize the cycle set $\cal C$, \DOMUS computes a \emph{cycle basis} covering all edges in the non-trivial biconnected components of the input graph $G$. This is done via a BFS from an arbitrary root to construct a tree $T$; for each non-tree edge $(u,v)$, we add the cycle formed by $(u,v)$ and the unique paths from $u$ and $v$ to their lowest common ancestor in $T$. This results in an initial set of size $|{\cal C}| = |E(G)| - |V(G)| + 1$.
For coordinate assignment, after computing $G_x$ and $G_y$, \DOMUS uses a compaction algorithm similar to that in \OGDF.
Finally, Glucose can make random choices. In order to enforce replicability of the experiments we forced it to work deterministically. 

\noindent\ref{li-adversary} For the motivations discussed in \cref{sec:introduction}, we compared the performance of \DOMUS against the TSM implementation available in \OGDF~\cite{chimani2013open}. Another option was to use as a benchmark one of the tools (see, e.g. \cite{Kieffer-2016}) that have been evaluated via user studies. However, this comparison would be useless, since most of this tools adopt a definition of orthogonal drawing which is different from the widely adopted one given in \cref{sec:preliminaries}. E.g., even if a vertex has degree less or equal than $4$ more than one of its incident edges can exit from the same direction. Also, vertices and bends neither have integer coordinates nor it is easy to ``snap'' them to a grid (see \cref{app:hola} for a comparison with these models).

\noindent\ref{li-set-of-graphs} We generated $4,100$ connected graphs uniformly at random, each with maximum vertex degree~$4$. Their sizes range from $20$ to $60$ vertices, and their densities -- defined as the ratio between the number of edges to the number of vertices -- span $1.25$ to $1.75$ in steps of 0.005. For degree $4$ graphs, the theoretical maximum density is~$2$.
Namely, for each $n=20,\dots,60$ and $i=1,2, \dots, 100$ we generated a graph $G_{n,i}$ with density $d_i = 1.25 + i\cdot(1.75-1.25)/100$. Hence, for each value of $n$ the graphs have roughly $100$ possible densities ranging from $1.25$ to $1.75$, in increments of 0.005. In terms of number of edges we have $|E(G_{n,i})| = \lfloor n\cdot d_i \rfloor$. For instance, the graph $G_{20,50}$ has $30$ edges, while $G_{60,100}$ has $105$ edges.
For each $n$ and $d_i$ we initialized a graph with $n$ vertices and no edges. Then, we repeatedly randomized two vertices $u$ and $v$ of the graph and added an edge between them, until $d_i$ was reached. If $u=v$ or if $(u,v)$ was already in the graph, or if either $u$ or $v$ was already degree $4$, we skipped the pair. Finally, the instance was rejected if non-connected.

Conducting experiments on a set of randomly generated graphs was a necessary choice, as, to the best of our knowledge, no publicly available collection of maximum-degree-4 graphs is currently recognized as a standard benchmark for graph drawing experiments.

\noindent\ref{li-computational-platform} All the experiments were performed running \OGDF and \DOMUS on a personal computer with a 3.15 GHz Intel processor (comparable with an M1 Apple processor).

\begin{figure}[tb!]
    \begin{subfigure}[b]{0.48\textwidth}
    \centering
        \includegraphics[width=\textwidth]{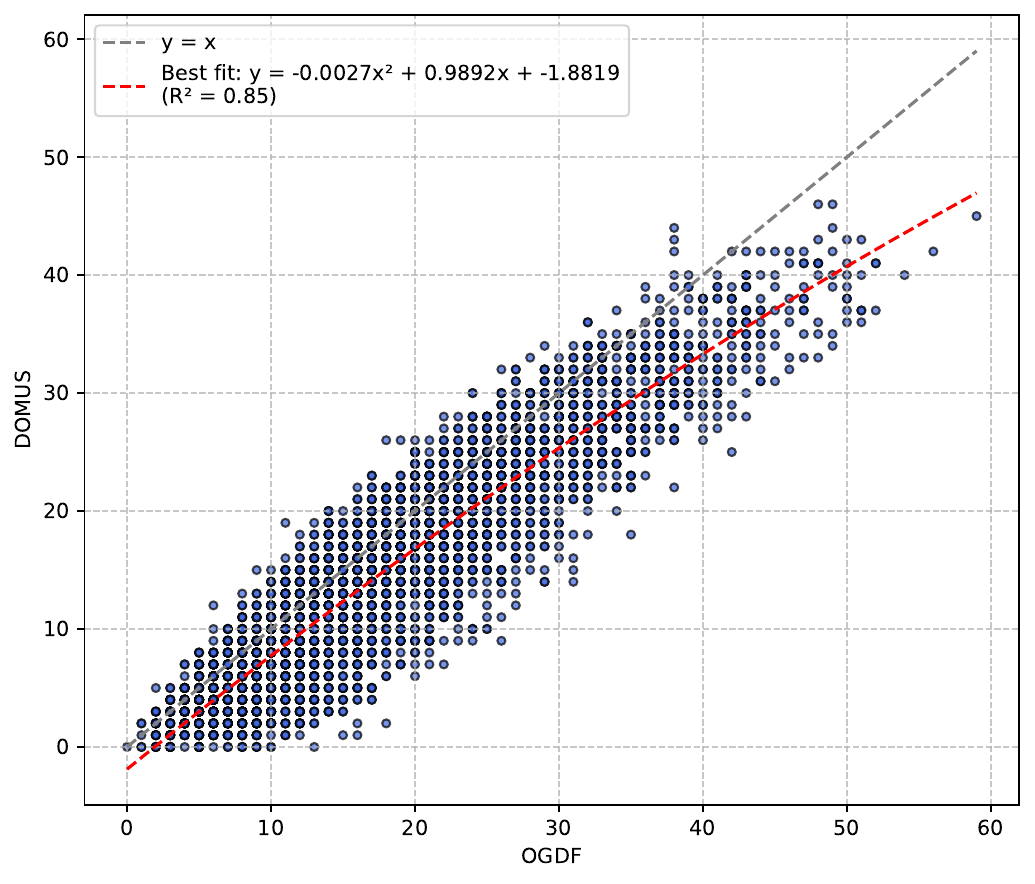}
        \subcaption{Bends.}
        \label{fig:bends_comparison}
    \end{subfigure}
    \hfill
    \begin{subfigure}[b]{0.48\textwidth}
    \centering
        \includegraphics[width=\textwidth]{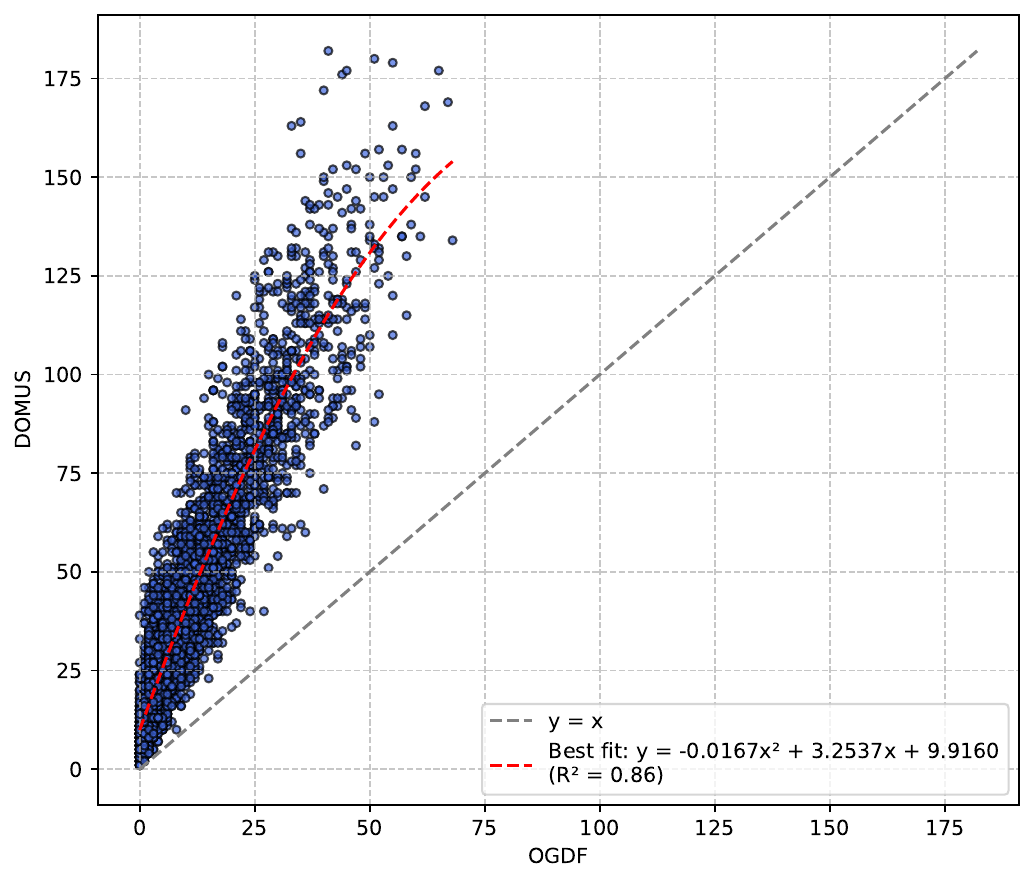}
        \subcaption{Crossings.}
        \label{fig:crossings_comparison}
    \end{subfigure}

    \par\bigskip
    
    \begin{subfigure}[b]{0.48\textwidth}
    \centering
        \includegraphics[width=\textwidth]{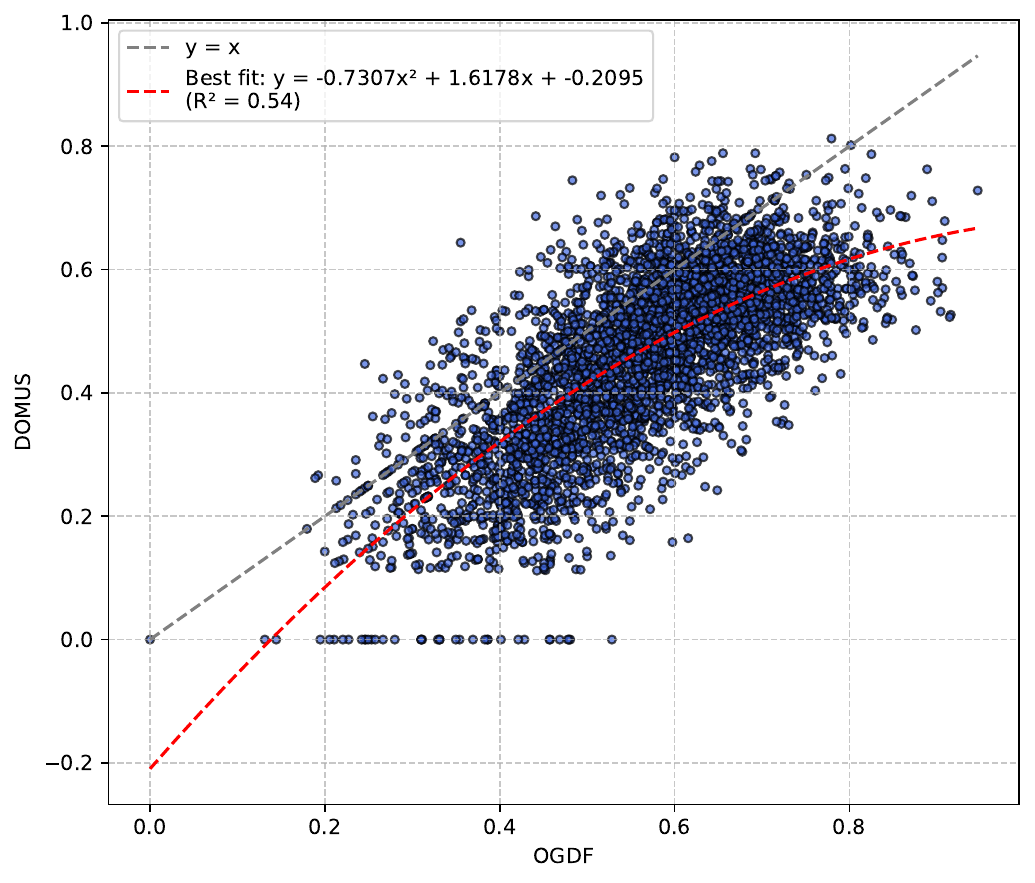}
        \subcaption{Bends Deviation.}
        \label{fig:bends_stddev_comparison}
    \end{subfigure}
        \hfill
    \begin{subfigure}[b]{0.48\textwidth}
    \centering
        \includegraphics[width=\textwidth]{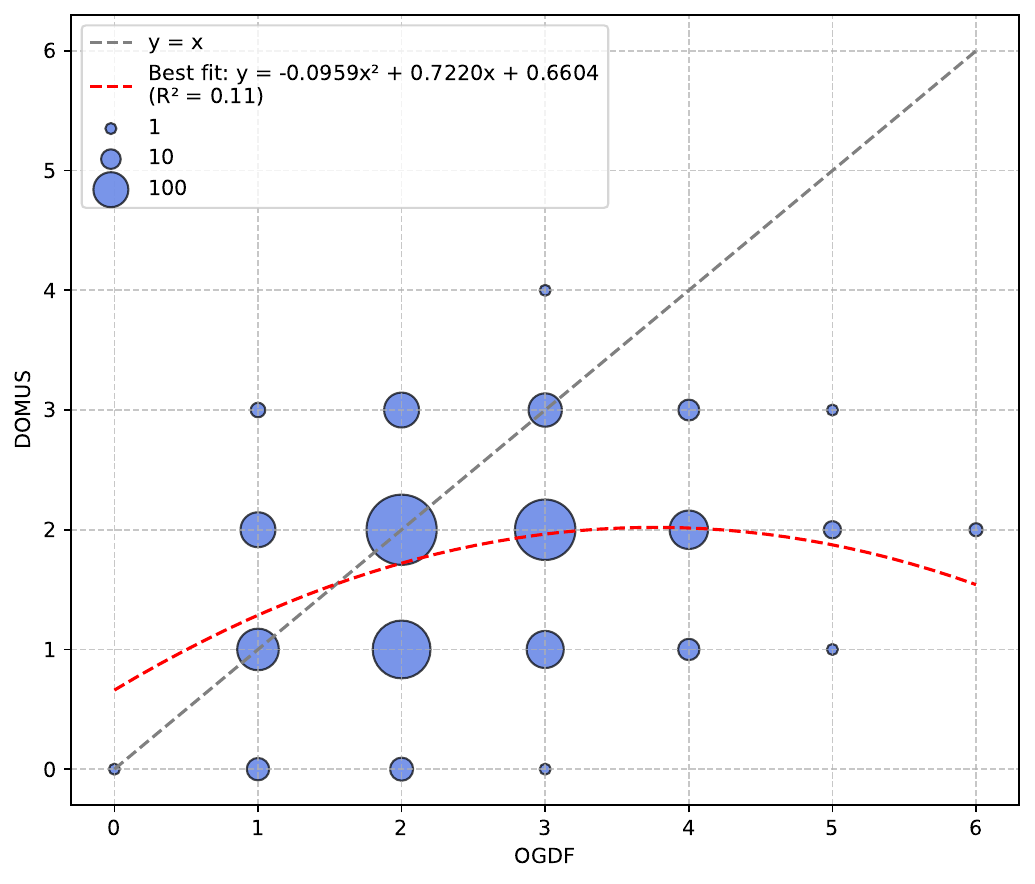}
        \subcaption{Max Bends.}
        \label{fig:max_bends_per_edge_comparison}
    \end{subfigure}

    \caption{Effectiveness of \OGDF and \DOMUS ``in vitro'': Bends and Crossings.}
    \label{fig:metrics-comparisons-part-1}
\end{figure}

\noindent\ref{li-metrics} Following the framework proposed in \cite{DBLP:journals/comgeo/WelzlBGLTTV97}, we compared the drawings produced by \DOMUS and those generated by \OGDF using the following metrics: 
the total number of {\bf Bends} in the drawing; 
the total number of edge {\bf Crossings};
the standard deviation of the number of bends per edge ({\bf Bends Deviation});
the maximum number of bends on any single edge ({\bf Max Bends});
the {\bf Area} occupied by the drawing; 
the sum of the length of all the edges ({\bf Total Edge Length}); 
the length of the longest edge ({\bf Max Edge Length}); 
the standard deviation of edge lengths ({\bf Edge Length Deviation}) and
the total computation {\bf Time} to build the drawing, measured in seconds. Using the terminology in \cite{DBLP:journals/cgf/BartolomeoCSPWD24}, we adopt most of the metrics suggested for ``quantitative individual'' and ``aggregated'' evaluations. A note on how the above metrics were computed is in \cref{sec:appendix-measures}.

\begin{figure}[tb!]        
    
    \begin{subfigure}[b]{0.48\textwidth}
    \centering
        \includegraphics[width=\textwidth]{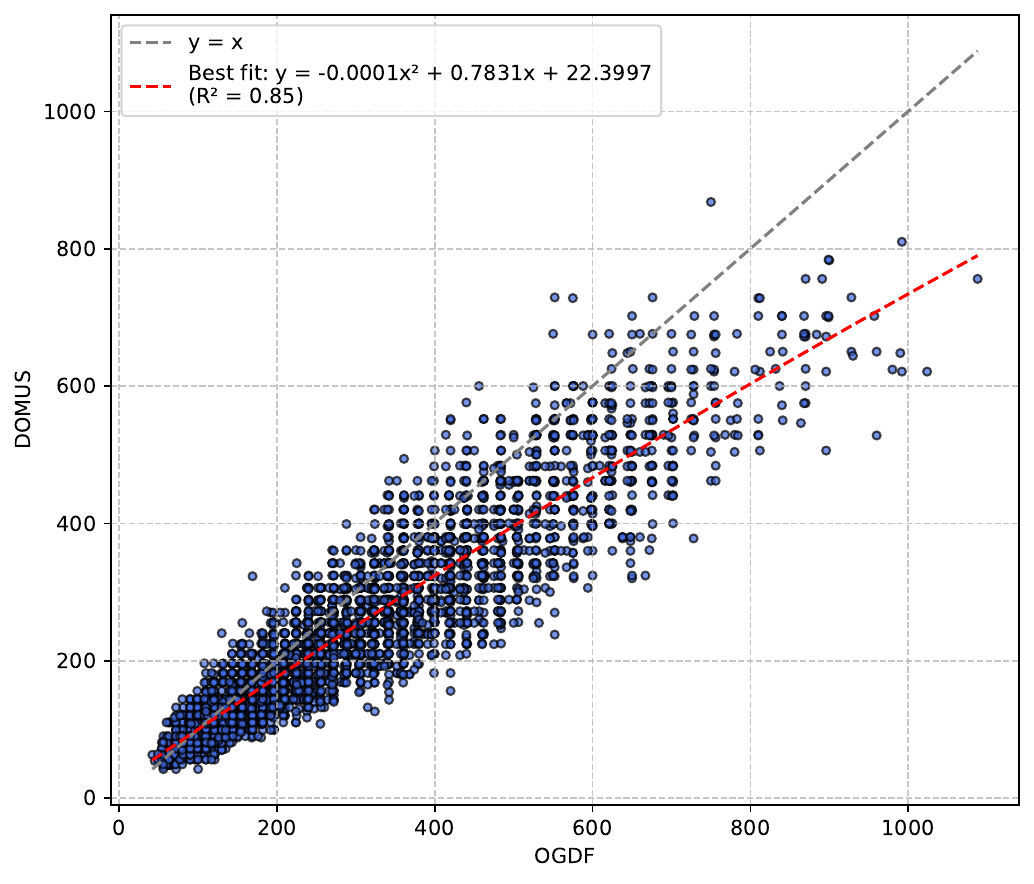}
        \subcaption{Area.}
        \label{fig:area_comparison}
    \end{subfigure}
            \hfill
    \begin{subfigure}[b]{0.48\textwidth}
    \centering
        \includegraphics[width=\textwidth]{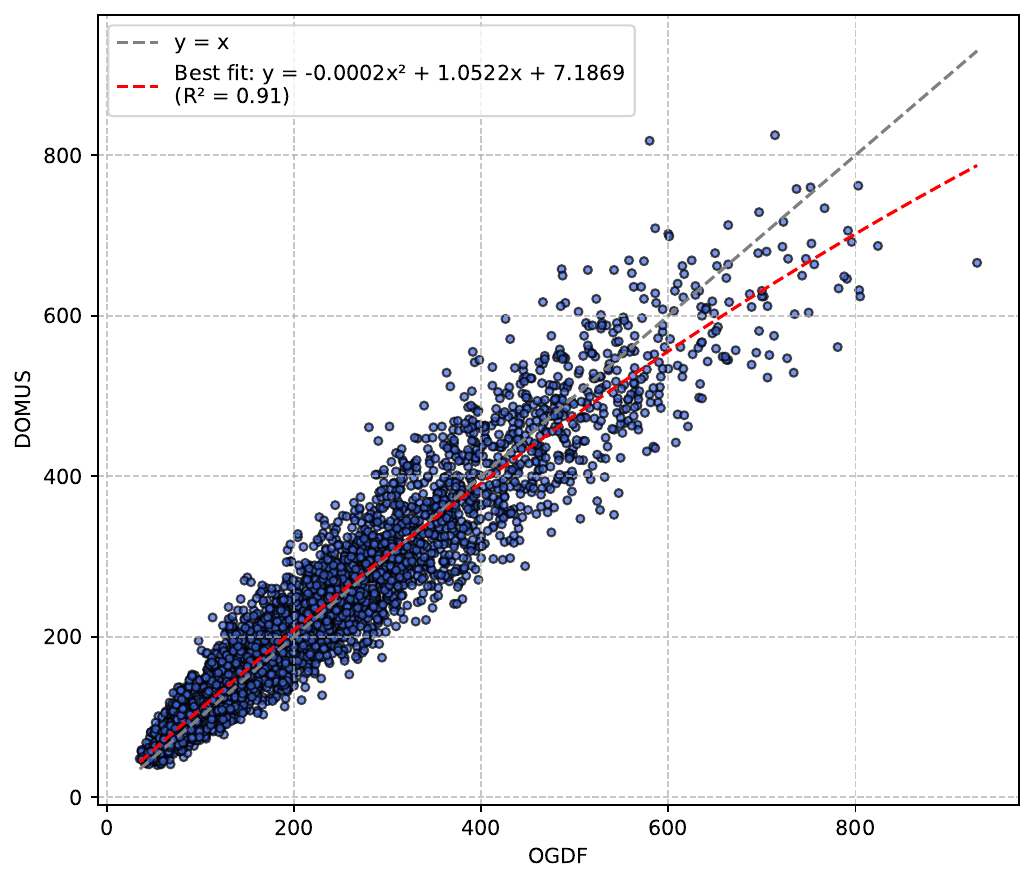}
        \subcaption{Total Edge Length.}
        \label{fig:total_edge_length_comparison}
    \end{subfigure}

    \par\bigskip
    
    \begin{subfigure}[b]{0.48\textwidth}
    \centering
        \includegraphics[width=\textwidth]{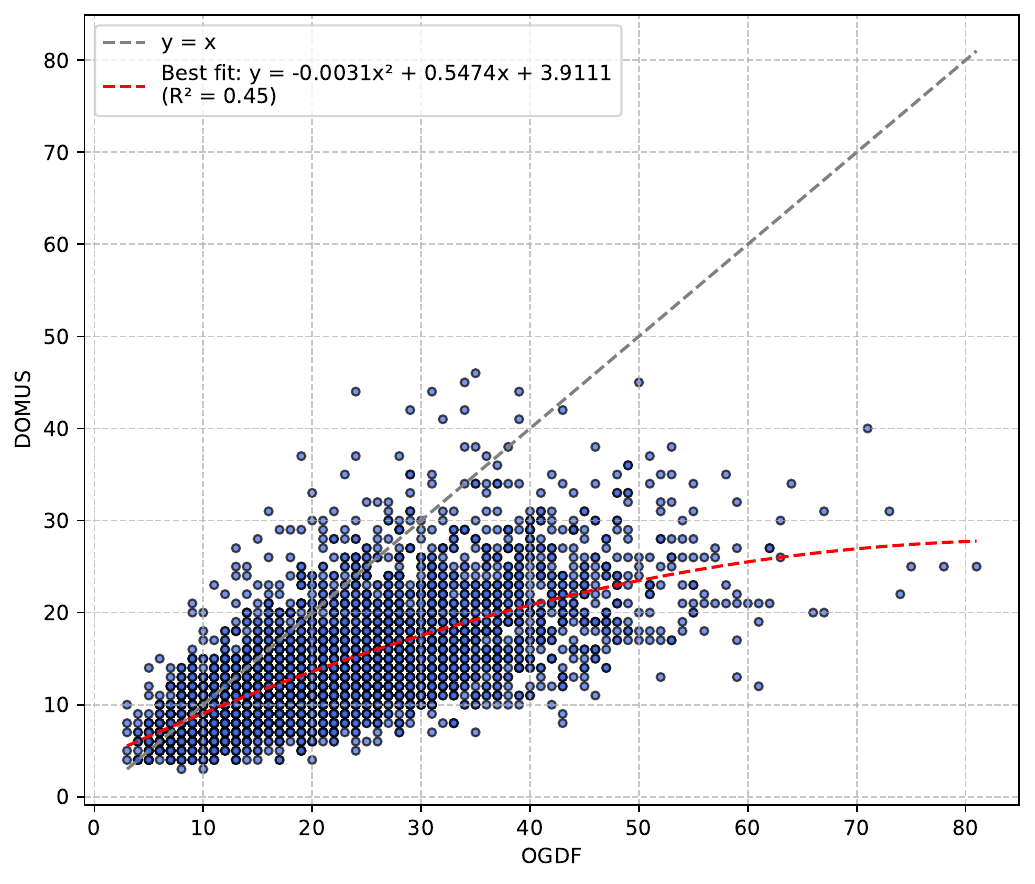}
        \subcaption{Max Edge Length.}
        \label{fig:max_edge_length_comparison}
        \end{subfigure}
            \hfill
    \begin{subfigure}[b]{0.48\textwidth}
    \centering
        \includegraphics[width=\textwidth]{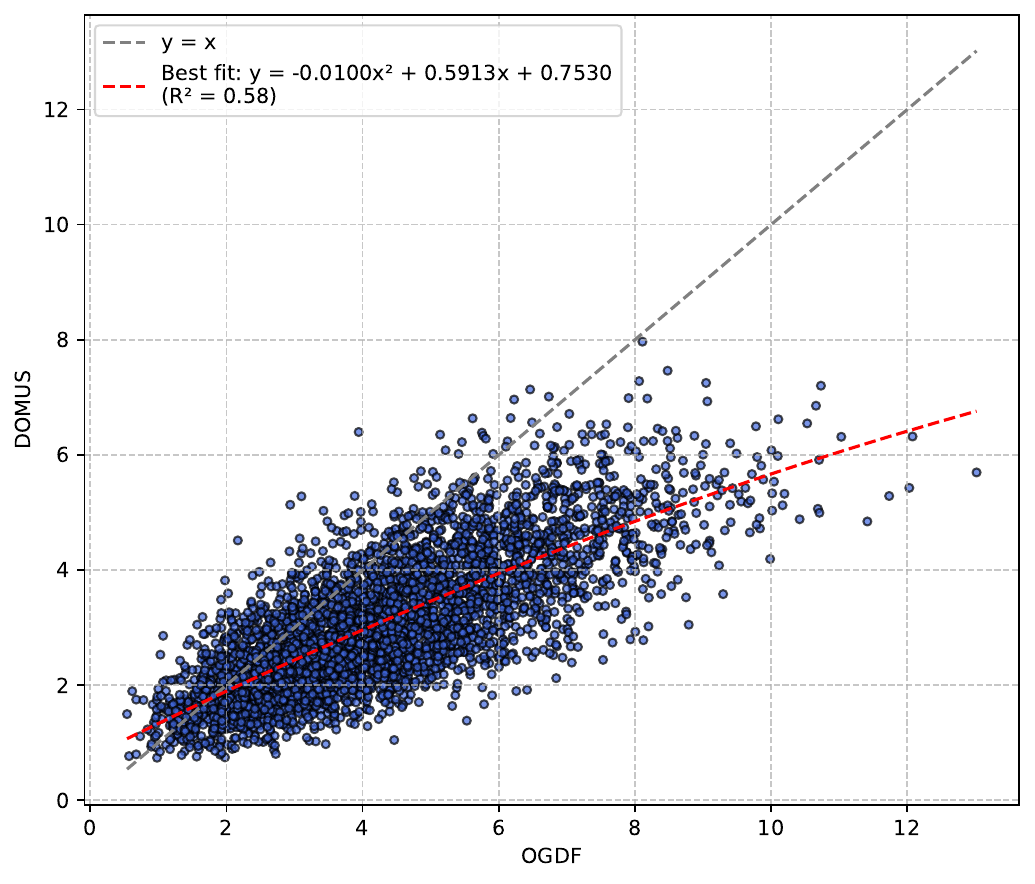}
        \subcaption{Edge Length Deviation.}
        \label{fig:edge_length_deviation_comparison}
    \end{subfigure}

    \caption{Effectiveness of \OGDF and \DOMUS ``in vitro'': Area and Edge Length.}
    \label{fig:metrics-comparisons-part-2}
\end{figure}

\noindent\ref{li-experimental-results} The results of our experiments are shown in \cref{fig:metrics-comparisons-part-1,fig:metrics-comparisons-part-2}, where each random graph is represented by a small circle. The $x$-coordinate corresponds to the metric value produced by \OGDF, while the $y$-coordinate reflects the value obtained by \DOMUS. Points below the bisector of the first quadrant (gray dashed line) indicate cases where \DOMUS outperformed \OGDF; points above indicate the opposite.
To account for overlapping data points (caused by graphs with same metric values), each plot includes a red dashed trend line summarizing the comparative performance. We also report the coefficient of determination ($R^2$, with $0 \leq R^2 \leq 1$), with values closer to $1$ indicating a better interpolation.

\cref{fig:bends_comparison} focuses on Bends. \DOMUS outperformed \OGDF on $78 \%$ of the graphs and matched its performance on another $7 \%$. The trend line indicates a modest quadratic improvement, a linear improvement of approximately $2 \%$, and a constant improvement of about $2.5$ bends per graph. Notably, $89$ of the random graphs admit a rectilinear drawing.
Conversely, as expected (\cref{fig:crossings_comparison}), \OGDF sharply outperforms \DOMUS in terms of Crossings. In this case, $149$ random graphs admit a planar drawing.
%
%
%
Additional insights regarding bends are in \cref{fig:bends_stddev_comparison,fig:max_bends_per_edge_comparison}, which report the Bends Deviation and the Max Bends per edge, respectively. For the former metric, \DOMUS outperforms \OGDF on $85\%$ of the graphs. For the latter, \DOMUS matches \OGDF on $50 \%$ of the graphs and outperforms it in the remaining $50\%$. Since the possible values of the maximum number of bends per edge are relatively few, many graphs correspond to the same point. Because of this, in \cref{fig:max_bends_per_edge_comparison} we represent with a variable size circle the number of graphs having the maximum number of bends per edge. The area of each circle is proportional to the corresponding number of graphs. These results indicate that our approach not only reduces the total number of bends but also achieves a more uniform distribution of bends across edges, thereby avoiding drawings with edges that contain long sequences of consecutive bends.


\DOMUS yields more compact drawings compared to those generated by \OGDF~(\cref{fig:area_comparison}). Specifically, our approach results in smaller drawing areas for $78 \%$ of the graphs, and matches \OGDF in $18 \%$ of the cases. The trend line indicates a consistent improvement, with an average area reduction of approximately $25 \%$.
%
%
We also achieve slightly better results than \OGDF in terms of edge length. \DOMUS produces drawings with lower Total Edge Length for $50 \%$ of the graphs and matches \OGDF on $1 \%$ of them (\cref{fig:total_edge_length_comparison}). More significantly, \DOMUS outperforms \OGDF on Max Edge Length (\cref{fig:max_edge_length_comparison}) for $87 \%$ of the graphs and obtain equal results for an additional $10 \%$. Similarly, for Edge Length Deviation (\cref{fig:edge_length_deviation_comparison}), \DOMUS performs better on $89 \%$ of the graphs. These results indicate that SM not only tends to produce slightly shorter edges overall but also ensures a more uniform edge length distribution, thereby avoiding disproportionately long edges that can hinder readability and spatial coherence.



\begin{figure}[tb!]
    \begin{subfigure}[b]{0.48\textwidth}
    \centering
        \includegraphics[width=\textwidth]{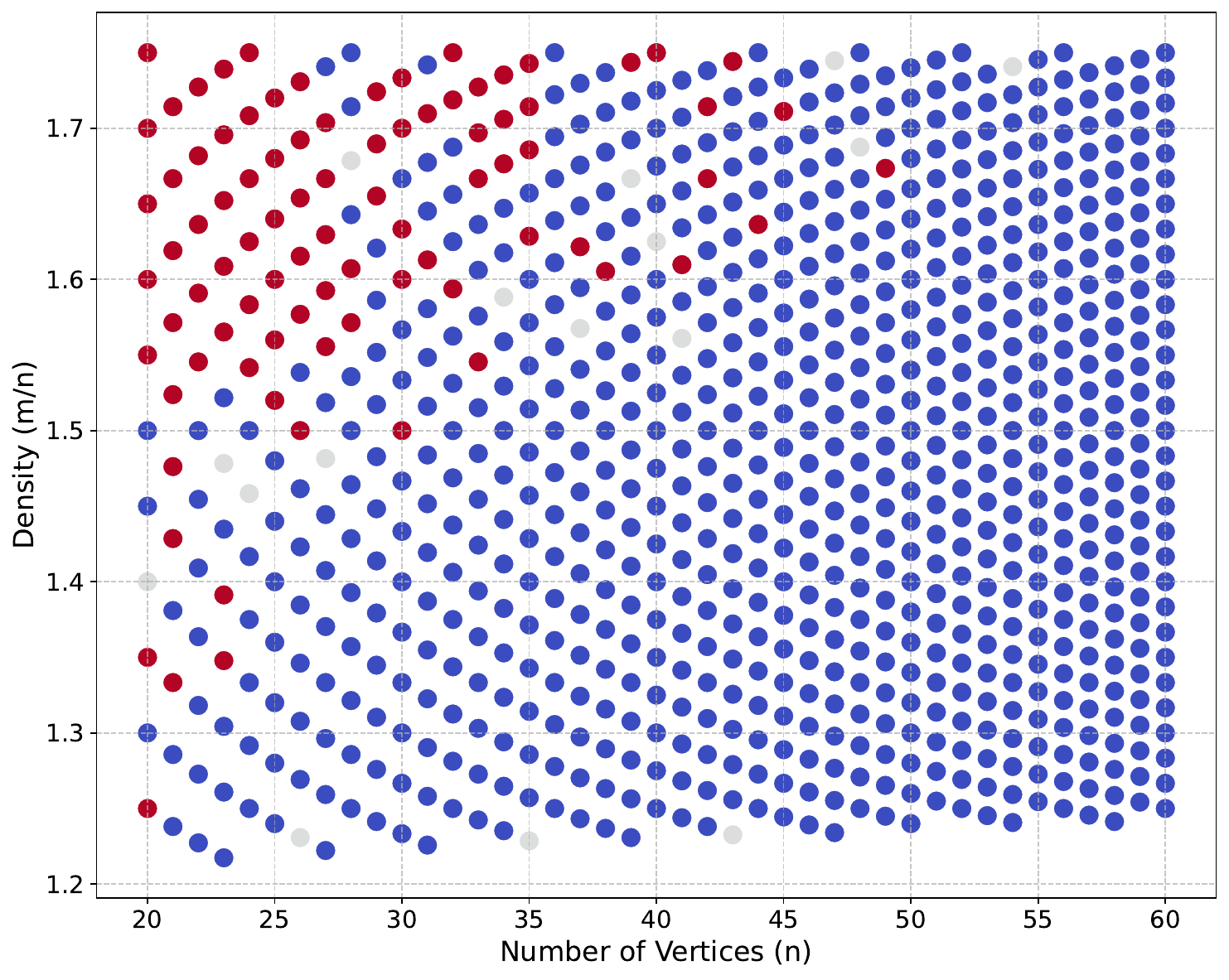}
        \subcaption{Bends}
        \label{fig:scatter_bends_comparison}
    \end{subfigure}
    \hfill
    \begin{subfigure}[b]{0.48\textwidth}
    \centering
        \includegraphics[width=\textwidth]{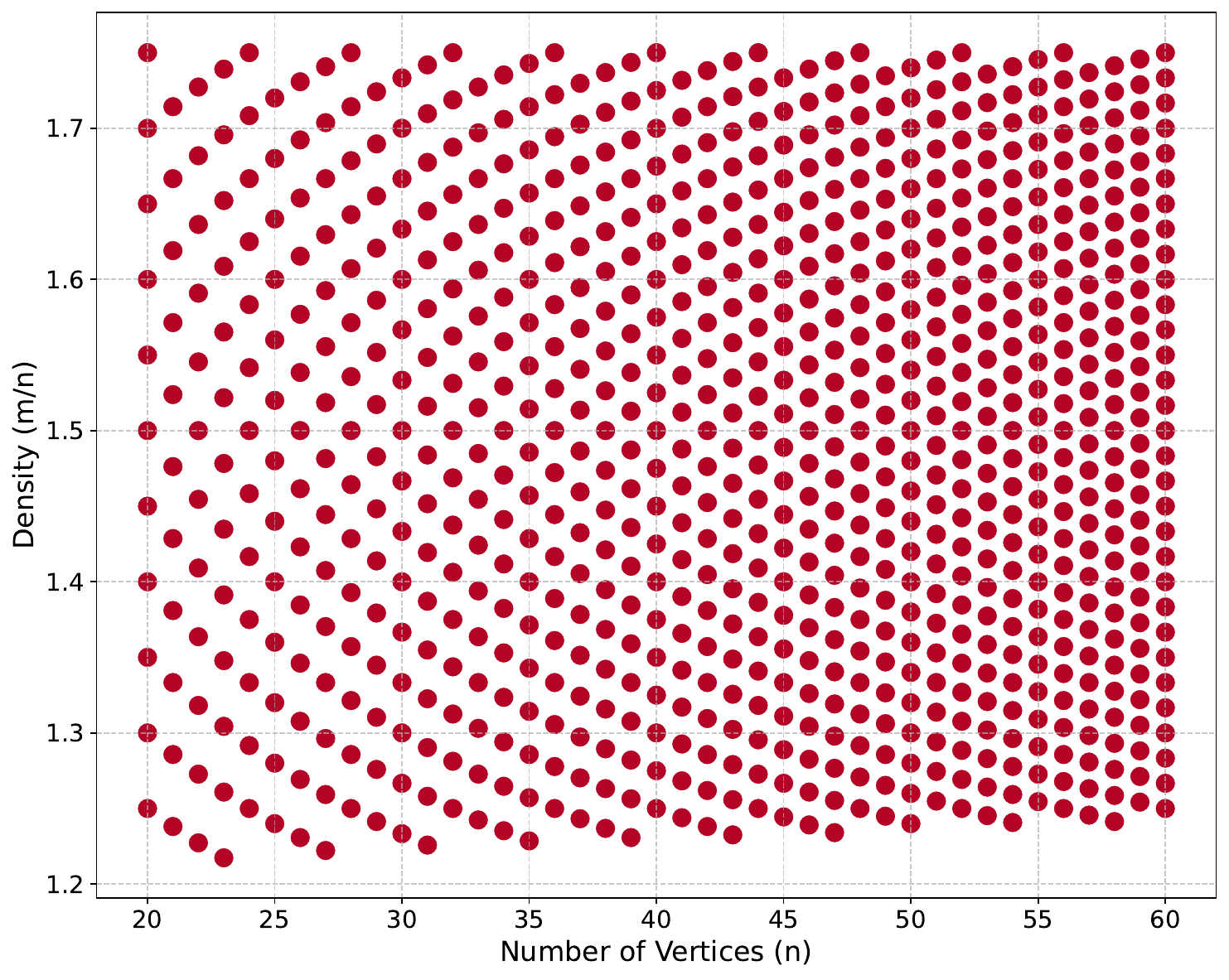}
        \subcaption{Crossings}
        \label{fig:scatter_crossings_comparison}
    \end{subfigure}

    \par\smallskip

    \begin{subfigure}[b]{0.48\textwidth}
    \centering
        \includegraphics[width=\textwidth]{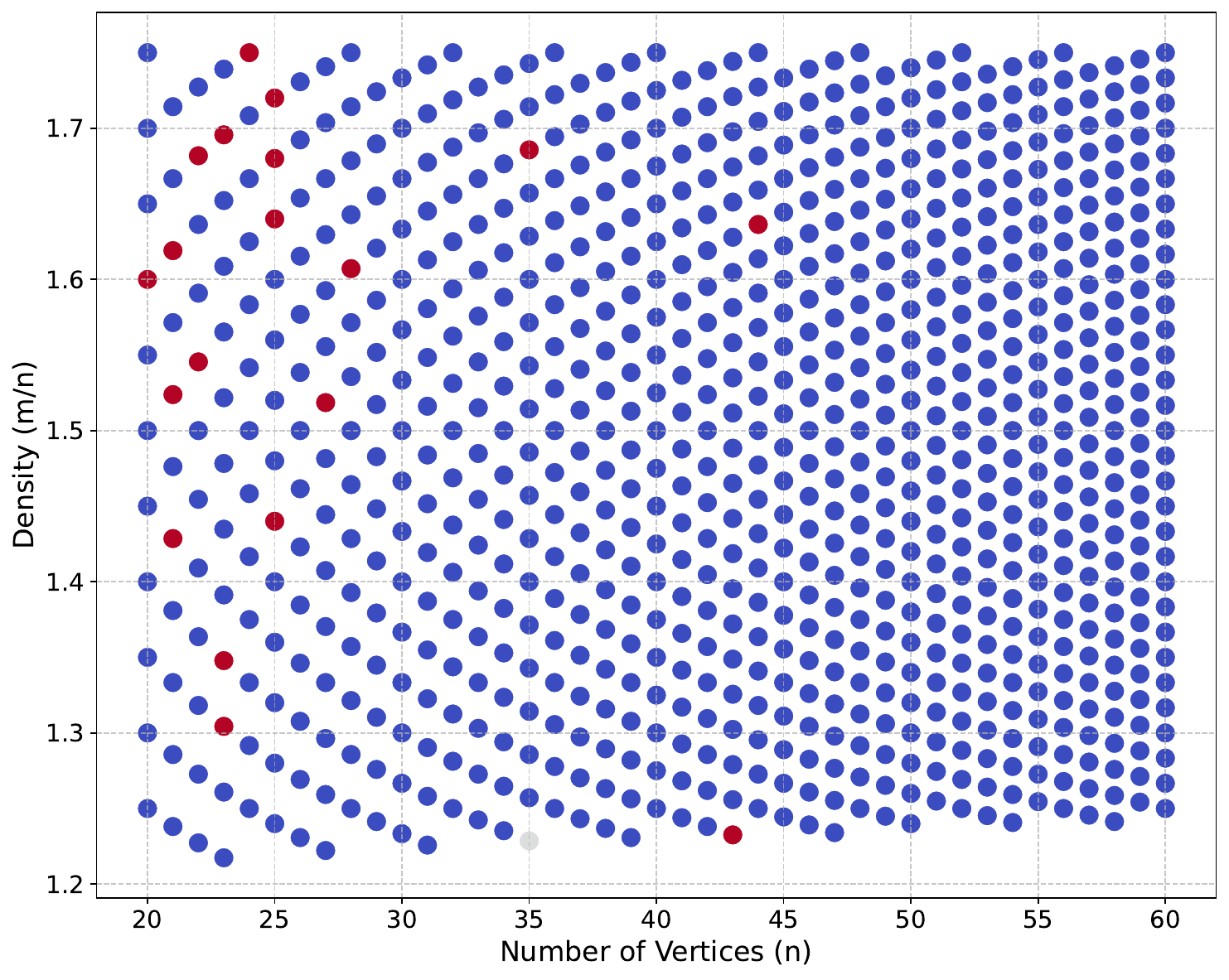}
        \subcaption{Bends Deviation}
        \label{fig:scatter_bends_stddev_comparison}
    \end{subfigure}
        \hfill
    \begin{subfigure}[b]{0.48\textwidth}
    \centering
        \includegraphics[width=\textwidth]{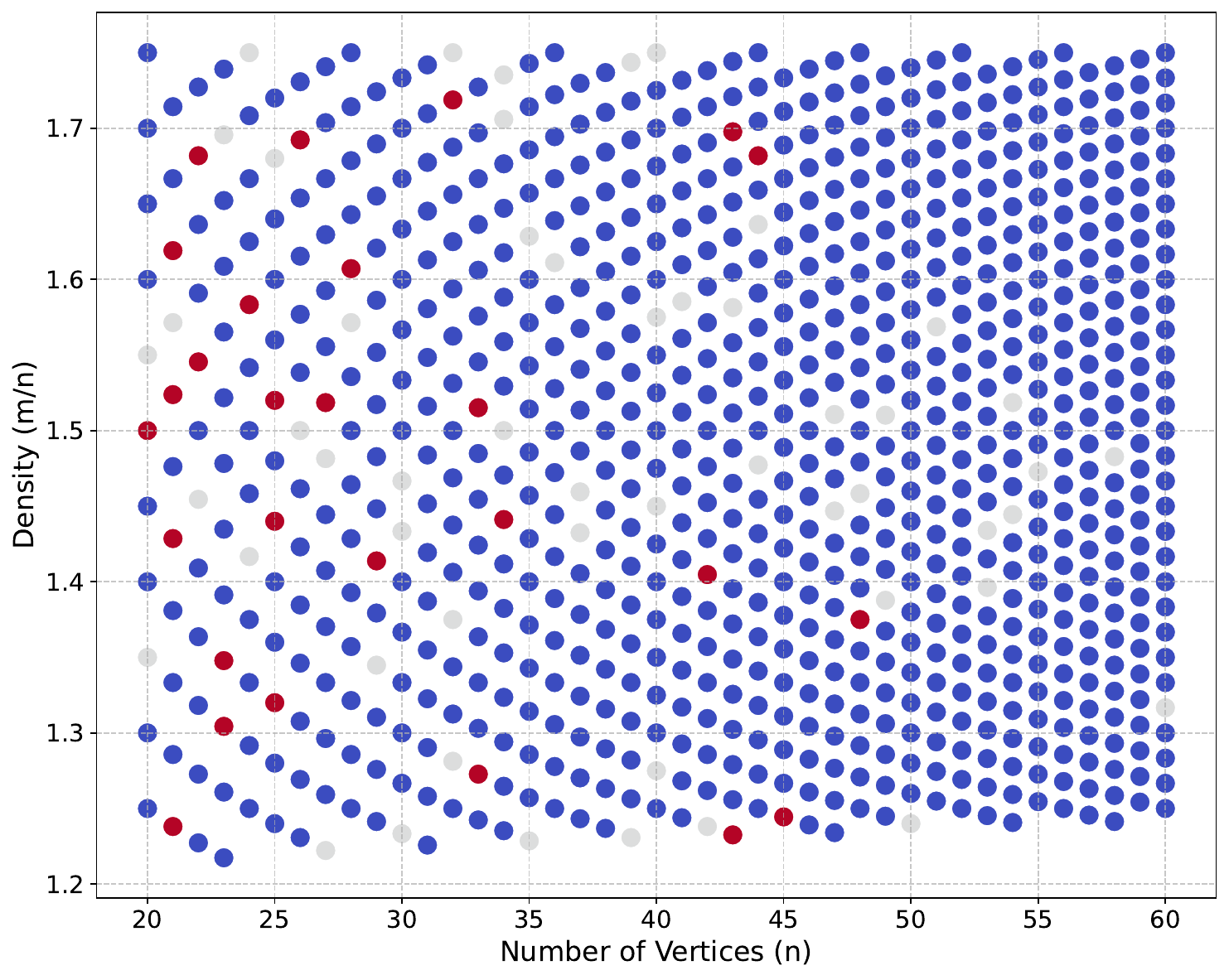}
        \subcaption{Max Bends}
        \label{fig:scatter_max_bends_per_edge_comparison}
    \end{subfigure}
    \caption{The effect of density: Bends and Crossings. Blue: \DOMUS is better. Red: \OGDF is better. Grey: parity.}
    \label{fig:scatter_comparisons-part-1}
\end{figure}

    
\cref{fig:scatter_comparisons-part-1} illustrates how the relative performance of \DOMUS and \OGDF varies with the number of vertices and the density of the input graphs. In each subfigure, a circle positioned at coordinates $(x, y)$ represents all graphs with $x$ vertices and density $y$. 
For a given metric $M$ and each pair $(x, y)$, let $M_O$ (respectively, $M_D$) denote the average value of $M$ for the drawings produced by \OGDF (respectively, \DOMUS) over all graphs with $x$ vertices and density $y$. The circle at $(x,y)$ is colored red, blue, or gray if $M_O < M_D$, $M_O > M_D$, or $M_O = M_D$, respectively.
\cref{fig:scatter_bends_comparison} shows that for small values of $n$, increasing graph density tends to favor \OGDF. However, this advantage gradually diminishes as $n$ increases. For the Area, we have a similar trend.
The other metrics do not show any clear dependency on either the number of vertices or the density
(\cref{fig:scatter_crossings_comparison,fig:scatter_bends_stddev_comparison,fig:scatter_max_bends_per_edge_comparison}). \cref{app:missing-figures-sec-5} contains additional diagrams about the density.

\begin{figure}[tb!]
    \hfill
    \begin{subfigure}[b]{0.35\textwidth}
    \centering
        \includegraphics[width=\textwidth]{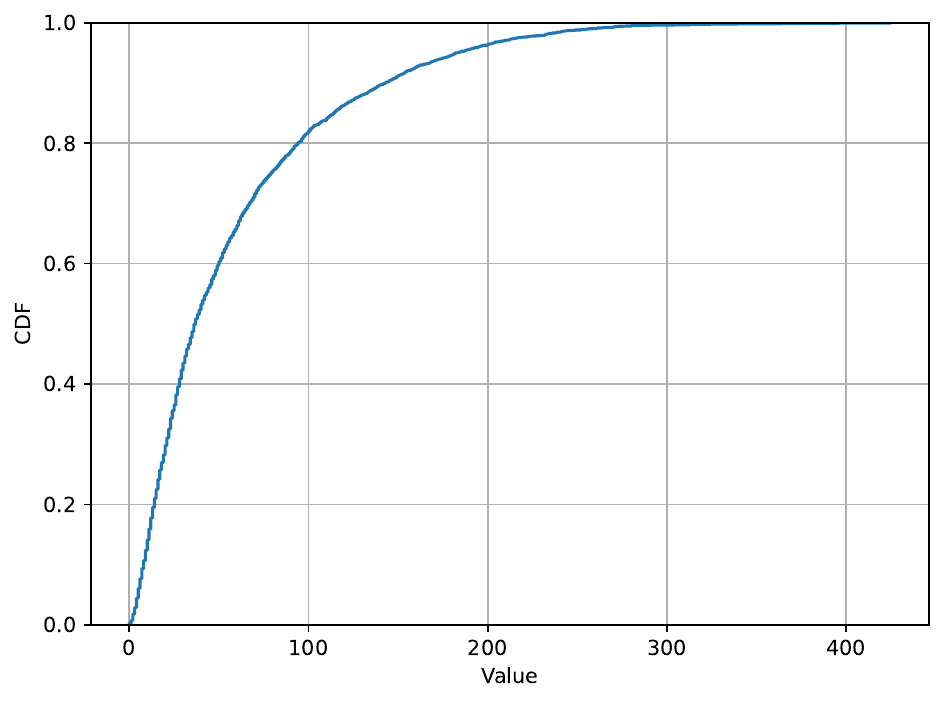}
        \subcaption{Cycles added during the computation.}  
        \label{fig:added-cycles-cdf}
    \end{subfigure}
    \hfill
    \begin{subfigure}[b]{0.35\textwidth}
    \centering
        \includegraphics[width=\textwidth]{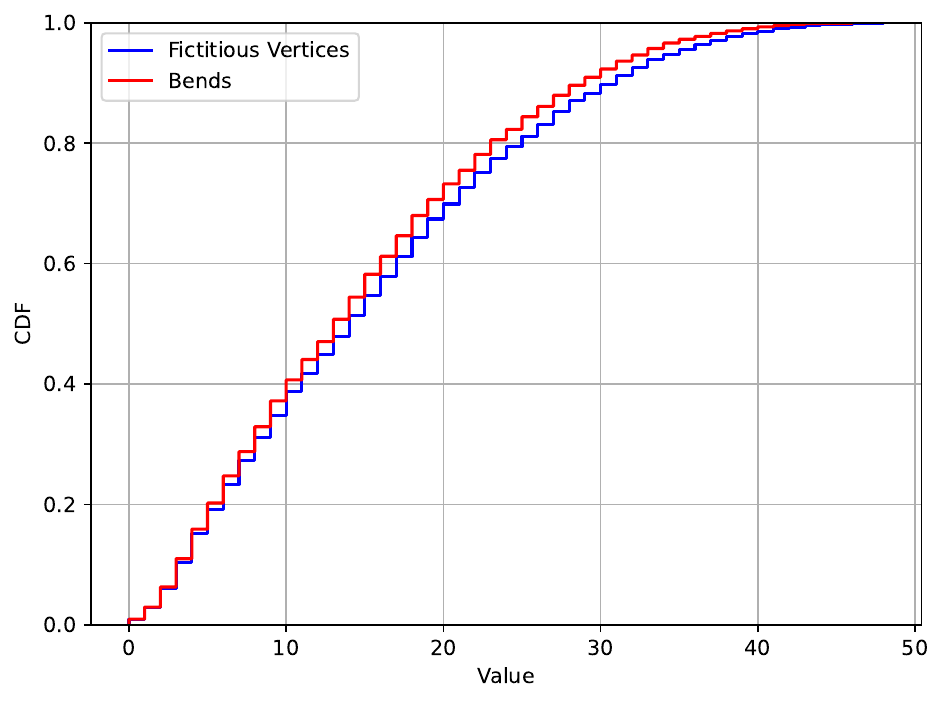}
        \subcaption{Dummy vertices and bends.}
        \label{fig:bends-cdf}
    \end{subfigure}
    ~~~~~~~~~

    \par\medskip

    \hfill
    \begin{subfigure}[b]{0.35\textwidth}
    \centering
        \includegraphics[width=\textwidth]{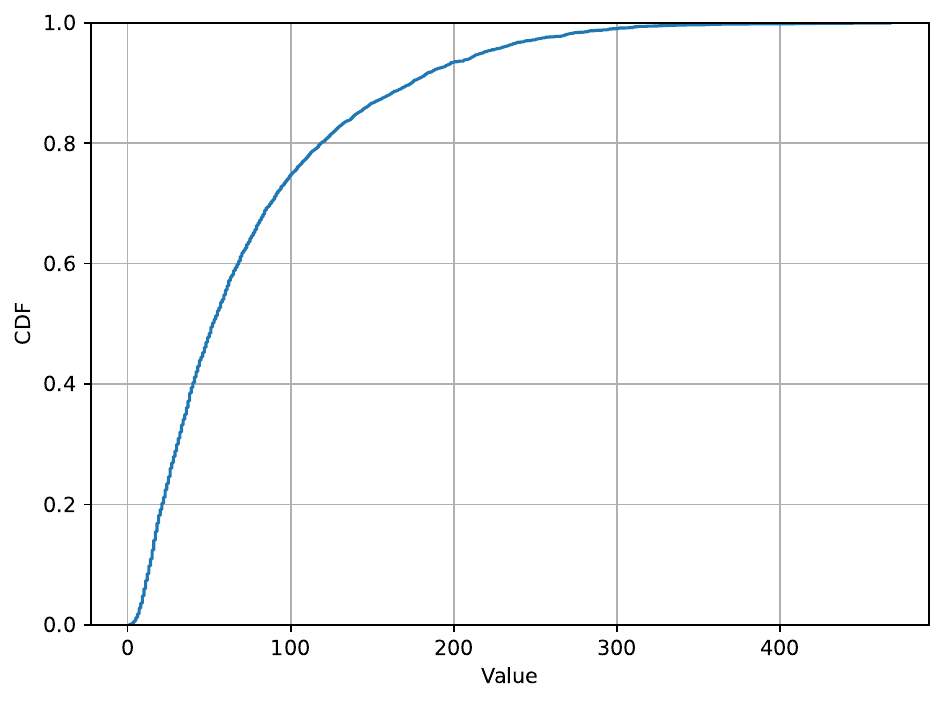}
        \subcaption{SAT invocations.}
        \label{fig:sat-invocations-cdf}
    \end{subfigure}
        \hfill
    \begin{subfigure}[b]{0.35\textwidth}
    \centering
        \includegraphics[width=\textwidth]{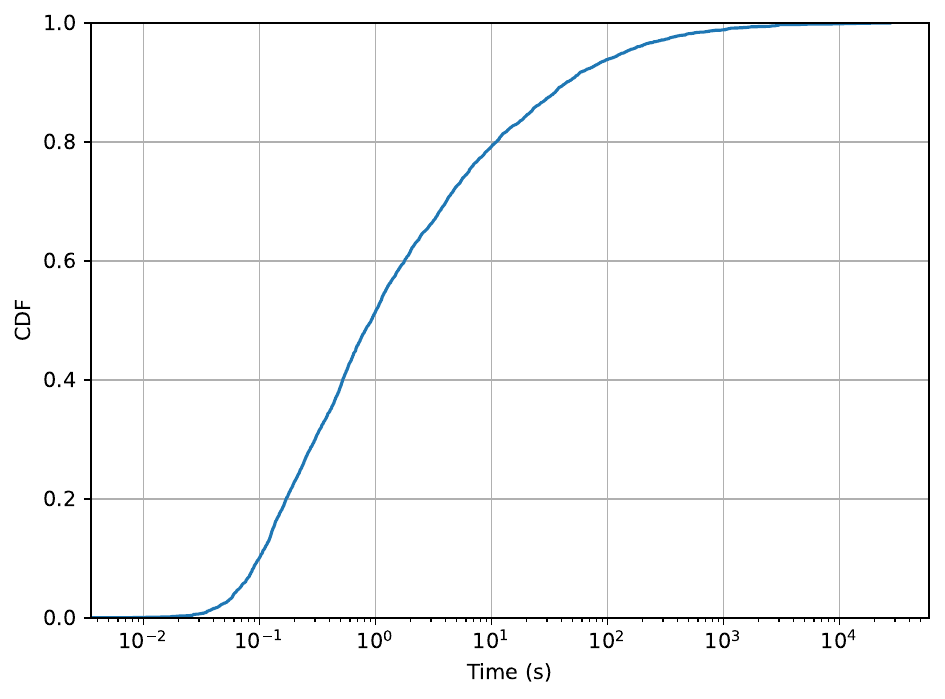}
        \subcaption{Time.}
        \label{fig:time-cdf}
    \end{subfigure}
    ~~~~~~~~~
    
    \caption{Several aspects of the \DOMUS computations.}
    \label{fig:internals}
\end{figure}

\noindent\ref{li-internals} To further characterize the computations performed by \DOMUS, for each run on a single graph, we measured the following (see \cref{fig:internals}).

First, we measured how many cycles were added to $\cal C$ during the computation. Each addition corresponds to a case where the current shape was found to be non-rectilinear drawable, triggering a new invocation of Shape Construction from Drawing Construction. In each of such invocations, the SAT solver was able to quickly find a satisfying assignment. The total number of Drawing Construction calls equals this number plus one. \cref{fig:added-cycles-cdf} shows a cumulative distribution function (CDF) of these data. For instance, in 80\% of the graphs, no more than 100 cycles were added to $\cal C$.

Second, we measured the number of dummy vertices introduced by the Shape Construction, corresponding to how often the Shape Construction called itself due to the absence of a shape satisfying the constraints. In each of these invocations, the formula given to the SAT solver was unsatisfiable.
\cref{fig:bends-cdf} shows that at most $25$ dummy vertices have been added for $80 \%$ of graphs.
Such vertices do not necessarily become bends in the final drawing, since some of them can get an angle of $\pi$ between their incident edges. So, a dummy vertex that does not become a bend can be viewed as a failure in the attempt to decrease the constraints of the drawing. Hence, we also computed how many dummy vertices actually became bends in the final drawing. This is an indirect measure of the effectiveness of using SAT unsatisfiability proofs to introduce bends. \cref{fig:bends-cdf} shows that almost all the dummy vertices became bends in the drawing.

Third, we measured the total number of SAT solver invocations, computed as the sum of the Shape Construction self-invocations, the Shape Construction calls by the Drawing Construction, plus one. \cref{fig:sat-invocations-cdf} shows that at most $120$ SAT invocations have been done for $80 \%$ of graphs. The similarity between \cref{fig:added-cycles-cdf} and \cref{fig:sat-invocations-cdf} shows that most of the SAT solver calls ($80 \%$) depend on the need to add a cycle to $\cal C$. 
%
Also, the average number of Boolean variables and clauses in the input to the SAT solver were $352$ and $1224$, respectively.

Fourth, \cref{fig:time-cdf} shows that computing a drawing for \DOMUS required at most $10$ seconds for $80 \%$ of graphs and that there are graphs that required more than $100$ seconds, with an average latency of $70$ seconds. \OGDF did its work in an average of $0.06$ seconds.

\section{An Experimental Evaluation on the Rome Graphs}\label{sec:different-settings}

With small modifications, it is possible to apply our methodology to construct orthogonal drawings of graphs with degree greater than $4$. In representing such graphs, we use the same convention introduced in \cite{DBLP:conf/gd/FossmeierK95}. Namely, the edges exiting the same side of a vertex can partially overlap (they are distanced by a very small amount) only for the first segment representing them. See \cref{fig:degree-greater-than-four}. The variation of the methodology works as follows.
First, we modify the ${\cal F}_{G, \cal C}$ formula in such a way that if a vertex $v$ has degree greater than $4$ more than one edge can enter $v$ from the same direction, while keeping the constraint that at least one edge enters $v$ from all the directions. After that, all the machinery that looks for a shape, including the usage of the SAT solver, stays the same.
Second, once a shape has been found, $v$ is temporarily expanded into a box (which will not be shown in the drawing) and the edges incident to the same side of $v$ are separated by introducing inside the box dummy vertices that become bends in the final drawing. After that, a metric is computed with the same technique as before. See drawings of Rome graphs in \cref{fig:degree-greater-than-four}.

\begin{figure}
    \centering
    \includegraphics[page=1,width=0.4\textwidth]{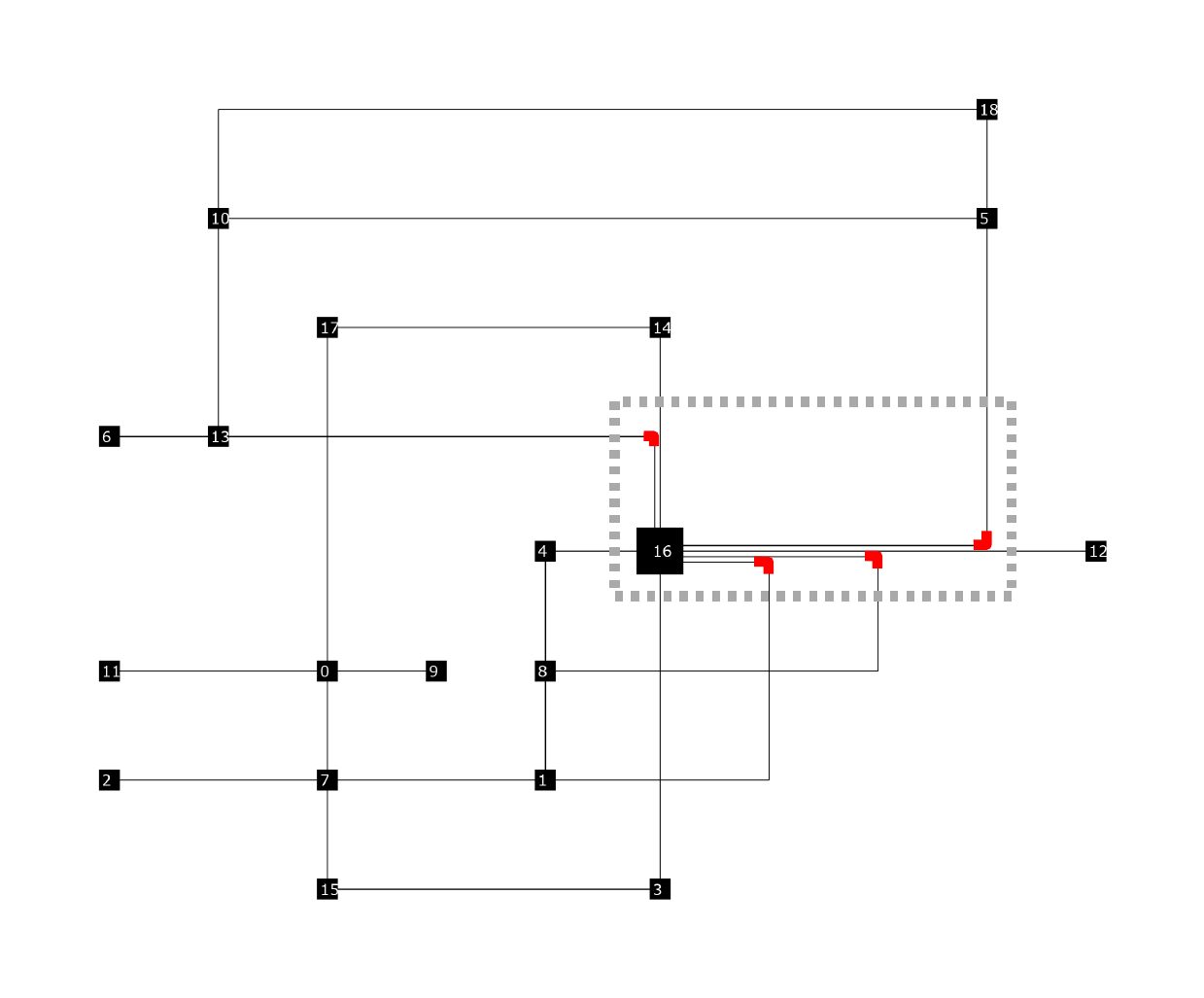}\hspace{0.4cm}
    \includegraphics[page=2,width=0.35\textwidth]{img-high-degree-final.pdf}
    \caption{Illustrations of how \DOMUS manages a vertex with degree greater than $4$. The dashed boxes show the expansions and the red segments show the bends that are added to the drawing.}
    \label{fig:degree-greater-than-four}
\end{figure}

We compared \OGDF and \DOMUS in experiments that we call ``in the wild'' using the Rome Graphs dataset \cite{DBLP:journals/comgeo/WelzlBGLTTV97}, which was originated by real-life applications and whose vertices have no degree restriction (see \cref{fig:rome-comparisons-part-1,fig:rome-comparisons-part-2}).
The experiments further increase the advantage of \DOMUS with respect to \OGDF (see \cref{sec:experiments-4}) in terms of Bends (\cref{fig:rome_bends_comparison,fig:rome_bends_stddev_comparison,fig:rome_max_bends_per_edge_comparison}) and confirm that \OGDF performs much better in terms of Crossings (\cref{fig:rome_crossings_comparison}). This probably depends on the low densities of the Rome Graphs. Further, \DOMUS performs better than \OGDF in terms of Area, Max Edge Length, and Edge Length Deviation (\cref{fig:rome_area_comparison,fig:rome_max_edge_length_comparison,fig:rome_edge_length_deviation_comparison}) and the two tools are comparable for Total Edge Length (\cref{fig:rome_total_edge_length_comparison}). Hence, the presence of vertices of degree greater than $4$ does not change too much the experimental results.
\OGDF was, even in this case, much faster than \DOMUS, since the average duration of its computations was $0.07$ seconds with a maximum of $3.3$ seconds. The \DOMUS computations took an average of $1.44$ seconds, with a maximum of $769$ seconds.

\begin{figure}[b!]
    \begin{subfigure}[b]{0.24\textwidth}
    \centering
        \includegraphics[width=\textwidth]{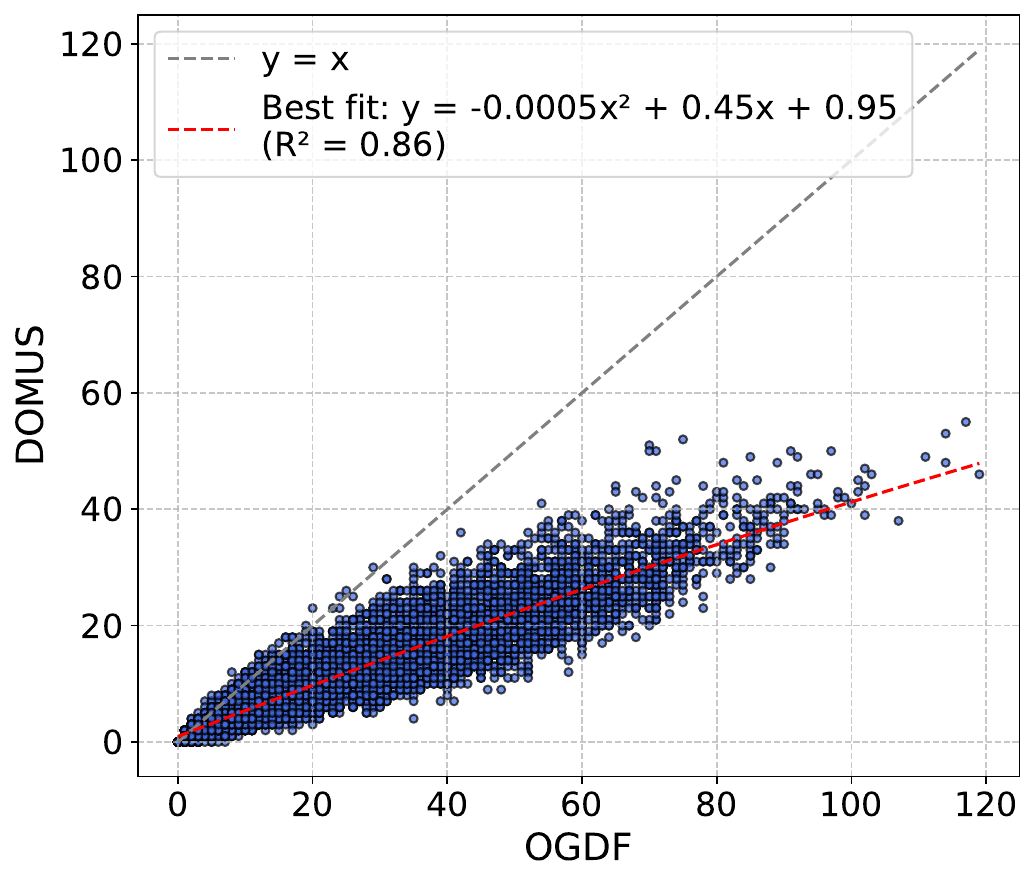}
        \subcaption{Bends.}
        \label{fig:rome_bends_comparison}
    \end{subfigure}
    \hfill
    \begin{subfigure}[b]{0.24\textwidth}
    \centering
        \includegraphics[width=\textwidth]{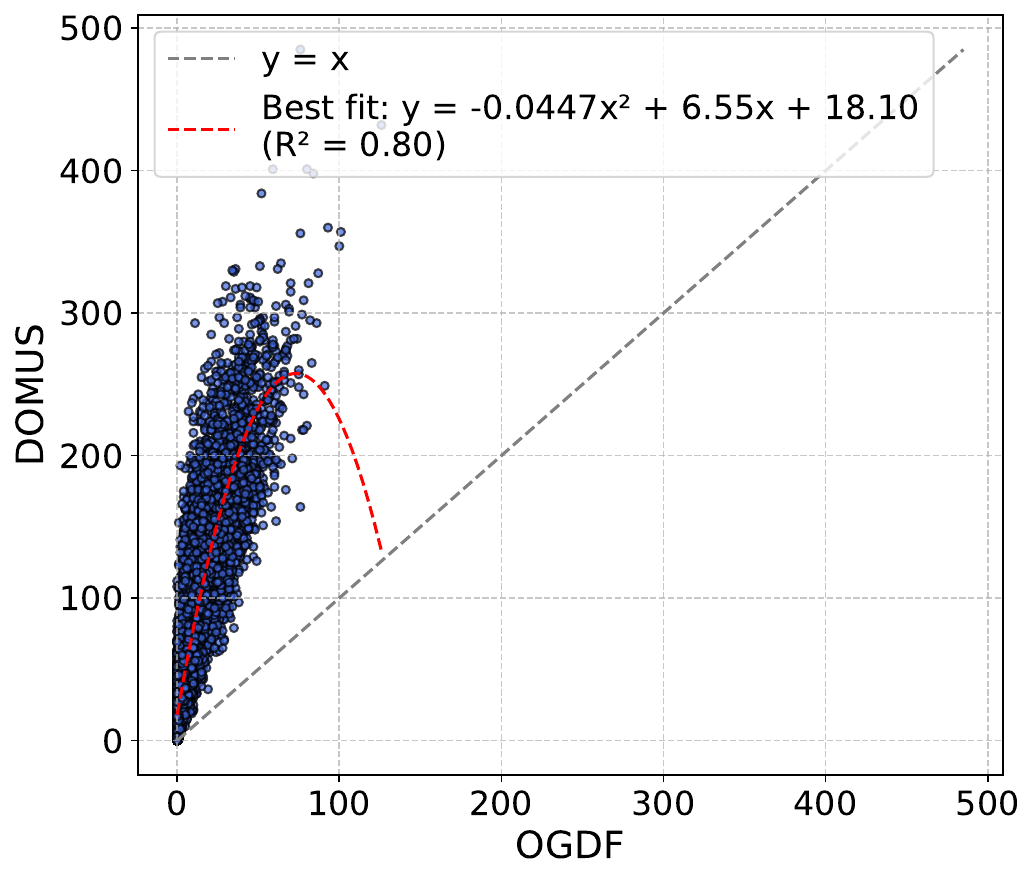}
        \subcaption{Crossings.}
        \label{fig:rome_crossings_comparison}
    \end{subfigure}
        \hfill    
    \begin{subfigure}[b]{0.24\textwidth}
    \centering
        \includegraphics[width=\textwidth]{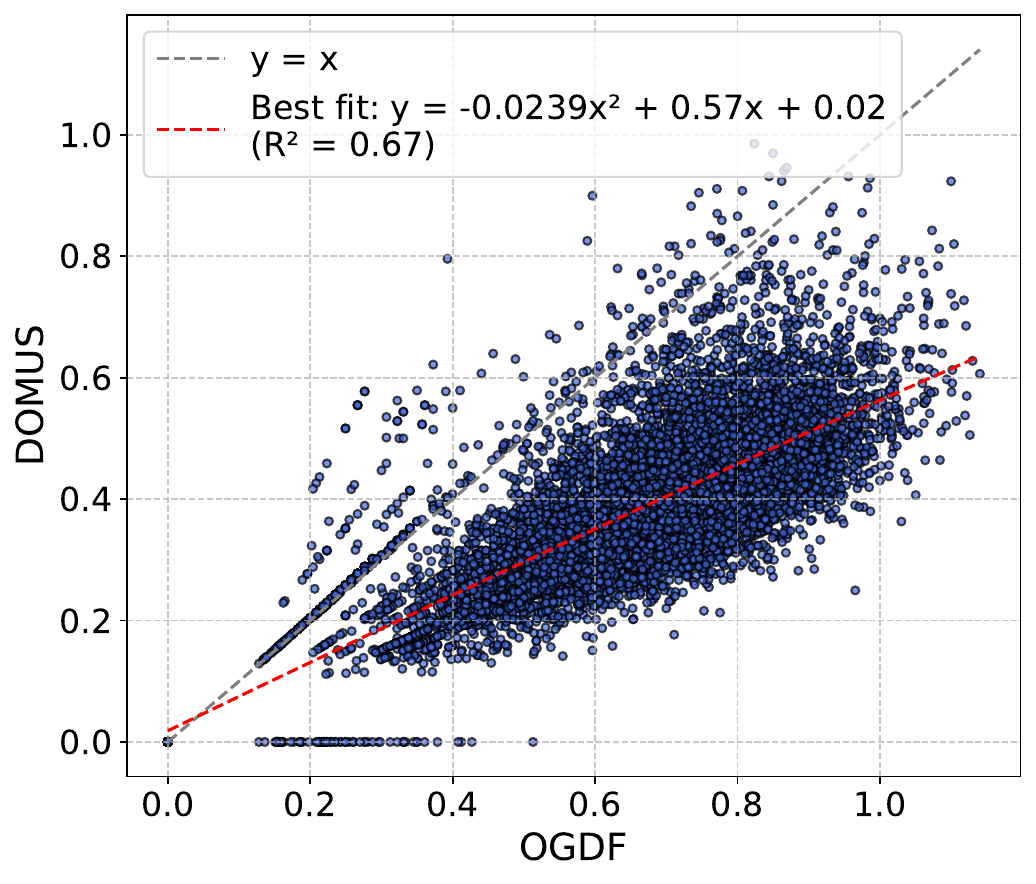}
        \subcaption{Bends Dev.}
        \label{fig:rome_bends_stddev_comparison}
    \end{subfigure}
        \hfill
    \begin{subfigure}[b]{0.24\textwidth}
    \centering
        \includegraphics[width=\textwidth]{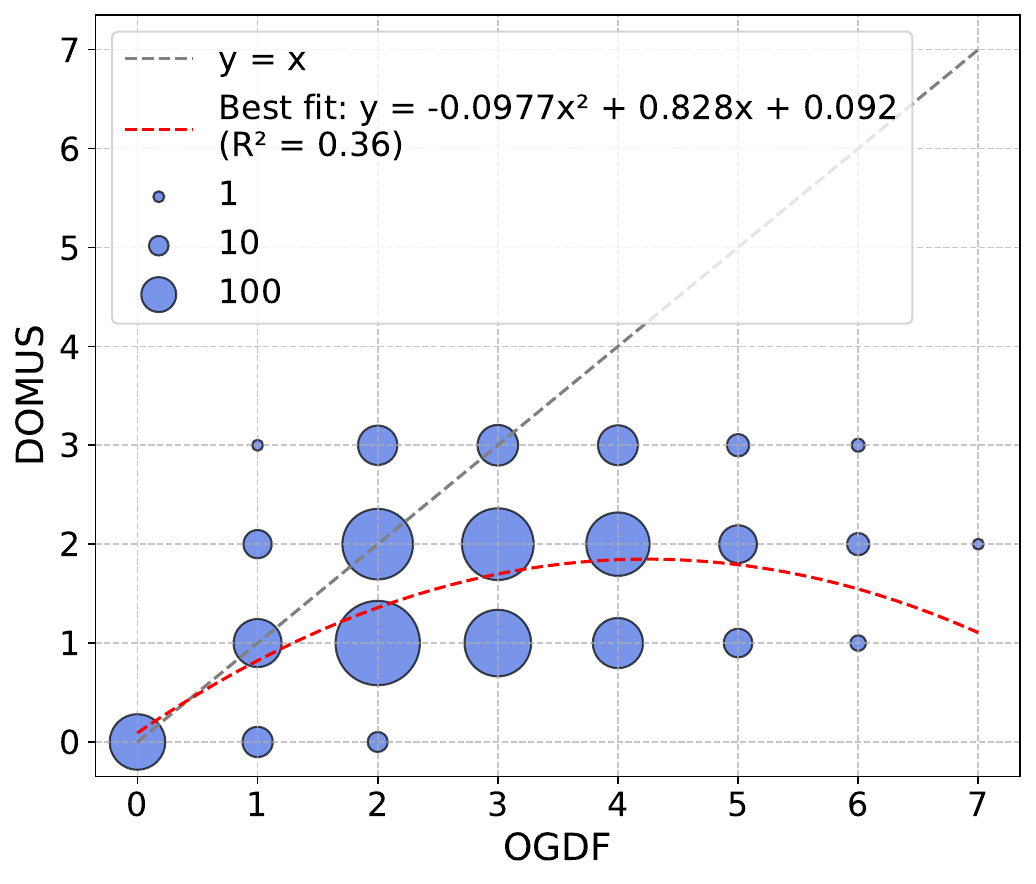}
        \subcaption{Max Bends.}
        \label{fig:rome_max_bends_per_edge_comparison}
    \end{subfigure}

    \caption{Effectiveness of \OGDF and \DOMUS ``in the wild'': Bends and Crossings.}
    \label{fig:rome-comparisons-part-1}
\end{figure}


    


\begin{figure}[tb!]        
    
    \begin{subfigure}[b]{0.24\textwidth}
    \centering
        \includegraphics[width=\textwidth]{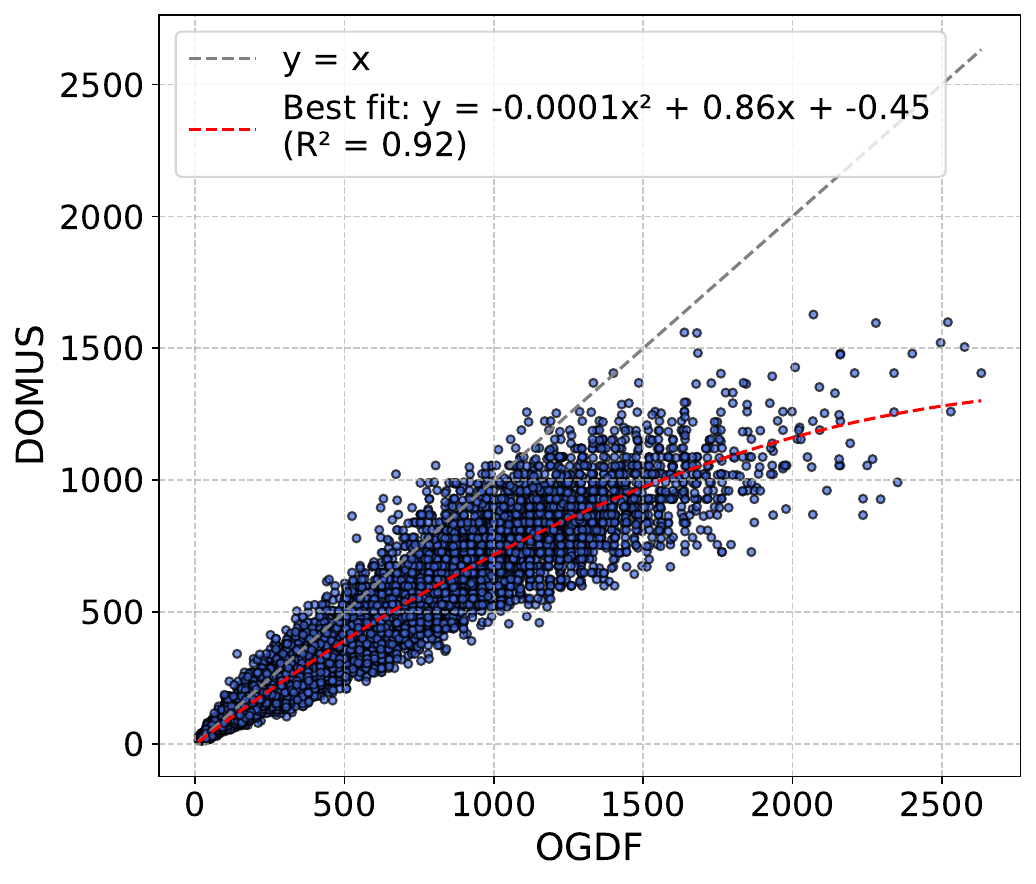}
        \subcaption{Area.}
        \label{fig:rome_area_comparison}
    \end{subfigure}
            \hfill
    \begin{subfigure}[b]{0.24\textwidth}
    \centering
        \includegraphics[width=\textwidth]{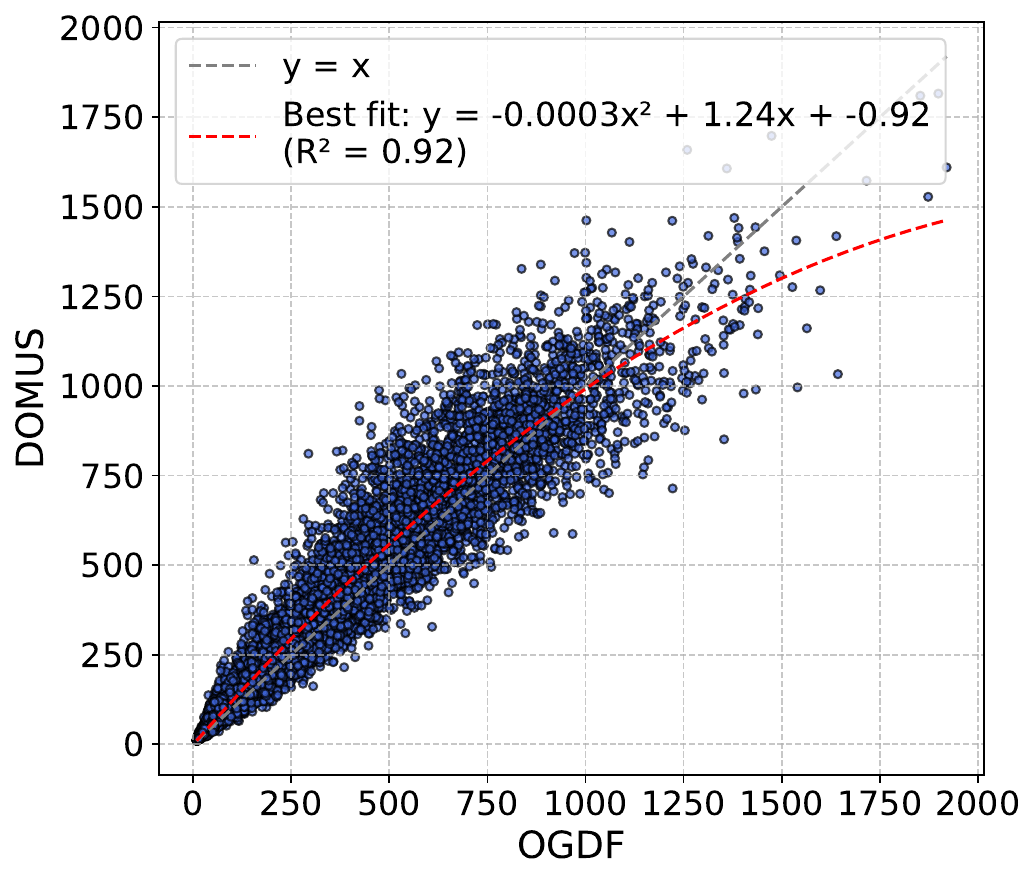}
        \subcaption{Total Edge L.}
        \label{fig:rome_total_edge_length_comparison}
    \end{subfigure}
            \hfill
    \begin{subfigure}[b]{0.24\textwidth}
    \centering
        \includegraphics[width=\textwidth]{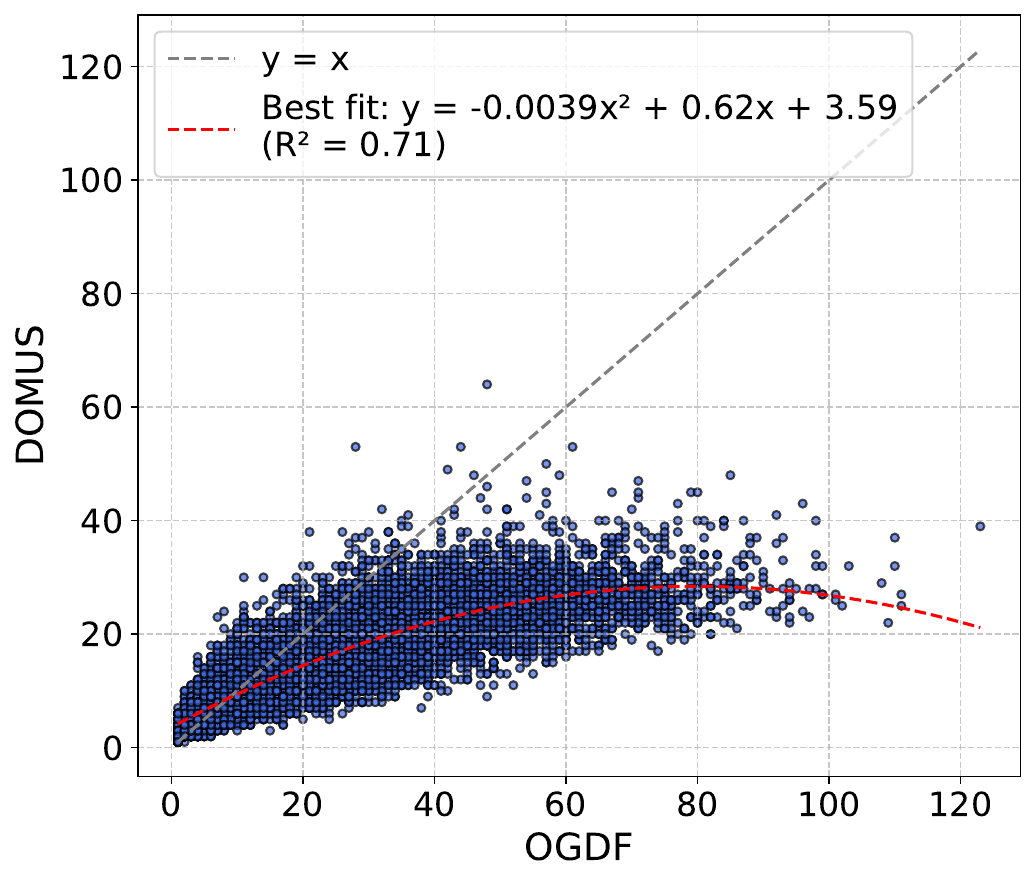}
        \subcaption{Max Edge L.}
        \label{fig:rome_max_edge_length_comparison}
        \end{subfigure}
            \hfill
    \begin{subfigure}[b]{0.24\textwidth}
    \centering
        \includegraphics[width=\textwidth]{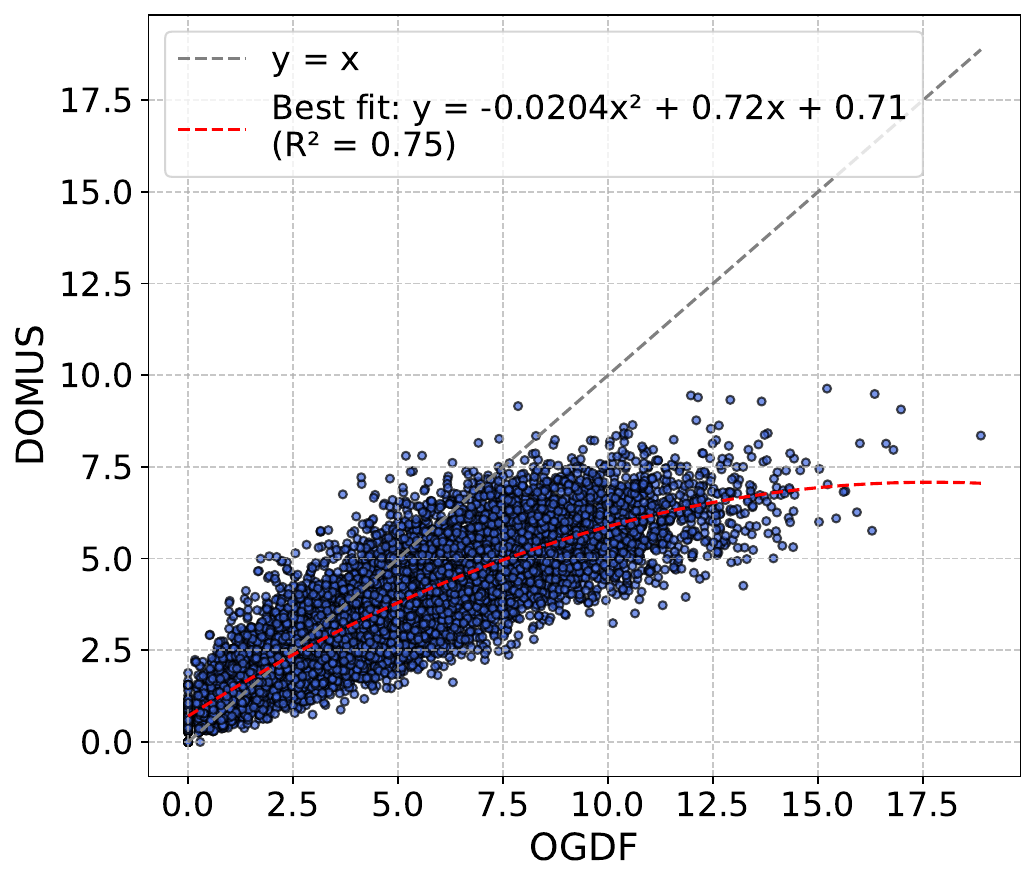}
        \subcaption{Edge L.\ Deviation.}
        \label{fig:rome_edge_length_deviation_comparison}
    \end{subfigure}
    \caption{Effectiveness of \OGDF and \DOMUS ``in the wild'': Area and Edge Length.}
    \label{fig:rome-comparisons-part-2}
\end{figure}

\section{Conclusions and Open Problems}\label{sec:conclusions}

About forty years after the introduction of the Topology-Shape-Metrics (TSM) approach for orthogonal drawings, we propose a novel methodology that subverts the TSM pipeline by prioritizing bend minimization over crossing minimization. Our experimental results demonstrate that this shift yields improvements across most standard metrics used in Graph Drawing to evaluate the effectiveness of a layout algorithm.

As a final remark, we observe that our methodology can be exploited to construct drawings incorporating several types of constraints (for an introduction to constrained orthogonal drawings, see, e.g., \cite{DBLP:conf/soda/EiglspergerFK00,DBLP:journals/constraints/Tamassia98}). If the graph contains edges that should be drawn upward it is easy to impose this by setting the value true for the corresponding $u$-variables of the ${\cal F}_{G, \cal C}$ formula. Similarly, giving the appropriate true-false values to the corresponding $\ell$-, $r$-, $d$-, and $u$-variables it is possible to specify that certain edges should enter their end-vertices respecting some specific orientations (so called side-constraints). Further, it is easy to specify that a given edge should be horizontal or vertical or to impose the shape of entire sub-graphs.

We open several promising research directions, including the following:
\begin{itemize*}
    \item Does there exist, at least for some families of graphs, a ``small'' set of cycles whose completeness implies the completeness of all cycles in $G$? If so, can such a set be efficiently computed? In \cref{app:small-set-of-cycles} it is shown that this problem is not trivial.
    \item Our methodology entirely disregards crossings. Is it possible to modify the approach to incorporate constraints that upper-bound the number of crossings?
    \item TSM and our methodology can be viewed as two extremes of a broader design space for constructing orthogonal drawings. Can hybrid methods be developed that combine elements of both approaches?
    \item In this paper, we employed compaction techniques originally developed within the TSM paradigm. 
    It would be interesting to design compaction methods that do not work on the planarized underlying graph, but that only preserve the shape of the edges.
    \item  Our approach can be leveraged to compute a maximal rectilinear drawable subgraph by iteratively invoking a SAT solver. Can the methodology be redesigned or optimized for greater efficiency when solving this specific problem?
\end{itemize*}

For the sake of reproducibility all our software and our data sets are publicly available at the anonymous repository \url{https://github.com/shape-metrics/domus}.


\clearpage

\bibliographystyle{plainurl}
\bibliography{bibliography}


\clearpage
\appendix

\section{Full Proofs from \cref{sec:shape-rectilinear-drawings} (\nameref{sec:shape-rectilinear-drawings})}\label{sec:missing-proofs}
\begin{lemma}\label{le:reduction}
Let $\Lambda = s_0, s_1, \dots, s_{p-1}$, with $p \geq 6$, be a circular sequence of labels from $\mathcal{L}$. Suppose that $\Lambda$ satisfies the following properties (where the indexes are intended $\text{mod~} p$):
\begin{enumerate}[label=\textbf{\alph*.}]
    \item For each $i \in \{0, \dots, p-1\}$ it holds that $s_i \neq s_{i+1}$ and $s_i \neq \overline{s_{i+1}}$;
    \item There exist distinct $i,j,k,l \in \{0, \dots, p-1\}$ such that the labels $s_i,s_j,s_k$, and $s_l$ are pairwise distinct.
\end{enumerate}
Then, there exist an index $h \in \{0, \dots, p-1\}$ such that removing the consecutive labels $s_{h}$ and $s_{h+1}$ from $\Lambda$ yields a new
sequence $\Lambda' = s_0, \dots, $ $s_{h-1}, s_{h+2}, \dots, s_{p-1}$, which also satisfies Properties~\textbf{\textit{a}} and \textbf{\textit{b}}.
\end{lemma}

\begin{proof}
Consider two consecutive symbols $s_h$ and $s_{h+1}$. Suppose that $s_h \in \{ U,D\}$; by Property~\textbf{\textit{a}} it follows that $s_{h+1} \in \{ L,R\}$, and vice-versa. Hence, since $\Lambda$ is circular and this alternation must hold for all consecutive pairs, it follows that $\Lambda$ has even length and consists of alternating labels, one from $\{R,L\}$ and the other from $\{U,D\}$. Now, if we remove $s_h$ and $s_{h+1}$ from $\Lambda$, the resulting sequence $\Lambda'$ preserves this alternation, and thus 
continues to satisfy Property~\textbf{\textit{a}}.

Concerning Property~\textbf{\textit{b}} we have two cases.
\begin{itemize}
    \item Case 1: The sequence $\Lambda$ contains a subsequence $s_h$, $s_{h+1}$, $s_{h+2}$, $s_{h+3}$ that satisfies Property~\textbf{\textit{b}}, i.e.\ $s_h = \overline{s_{h+2}}$ and $s_{h+1} = \overline{s_{h+3}}$. As the length of $\Lambda$ is at least $6$, the sequence $\Lambda'$ obtained from $\Lambda$ by removing $s_{h+4}$ and $s_{h+5}$ then satisfies Property~\textbf{\textit{b}}.

    \item Case 2: Suppose that $\Lambda$ does not contain such a 4-element subsequence as described in Case 1. Then there exists an index $h$ such that $s_h$ = $s_{h+2}$; let $\alpha=s_h$, and $\beta=s_{h+1}$; so $\Lambda$ contains the subsequence $\alpha$, $\beta$, $\alpha$.
    If $s_{h+3}$ and $s_{h+1}$ are also equal, we have that $\Lambda$ contains the subsequence $\alpha$, $\beta$, $\alpha$, $\beta$. Since the two labels are both repeated, the sequence $\Lambda'$, obtained by removing $s_h$ and $s_{h+1}$ from $\Lambda$, satisfies Property~\textbf{\textit{b}}.
    Otherwise, consider $s_{h+3} = \overline{s_h+1}$. We then have that $\Lambda$ contains the subsequence $\alpha$, $\beta$, $\alpha$, $\overline{\beta}$. Since we are not in Case 1, $\Lambda$ contains the subsequence $\alpha$, $\beta$, $\alpha$, $\overline{\beta}$, $\alpha$. Whatever the value after this subsequence is, $\beta$ or $\overline{\beta}$, for both the resulting subsequences 
    $\alpha$, $\beta$, $\alpha$, $\overline{\beta}$, $\alpha$, $\beta$
    and
    $\alpha$, $\beta$, $\alpha$, $\overline{\beta}$, $\alpha$, $\overline{\beta}$
    we have that the last two labels are repeated elsewhere in the subsequence, meaning they can be removed from $\Lambda$ without falsifying Property~\textbf{\textit{b}}.
\end{itemize} \end{proof}

\thCycle*

\begin{proof}
First, we prove that if $c$ is rectilinear drawable it is complete.
Consider any drawing $\Gamma$ of $c$ and consider the shape $\lambda$ associated with $\Gamma$. Let $\ell_h$ be a horizontal line that intersects $\Gamma$, but does not pass through any point representing a vertex of $c$, and hence does not intersect any horizontal segment of $\Gamma$. Consider any vertical segment intersected by $\ell_h$ and suppose that it represents edge $(u,v)$. If the point representing $u$ lies below $\ell_h$ and the point representing $v$ lies above $\ell_h$ we have that $\lambda(u,v)=U$, otherwise we have that $\lambda(v,u)=U$. Next, traverse the cycle $c$ starting from $v$, proceeding in the direction opposite to $u$. Let $(w,z)$ be the first edge intersected by $\ell_h$ after $(u,v)$. If the point representing $w$ lies above that of $z$, then we have $\lambda(w,z)=D$, otherwise $\lambda(z,w)=D$.
Similarly, let $\ell_v$ be a vertical line that intersects $\Gamma$, but does not pass through any point representing a vertex of $c$, and hence does not intersect any vertical segment of $\Gamma$. Consider a horizontal segment intersected by $\ell_v$ and suppose that it represents edge $(u,v)$. If the point representing $u$ lies to the left of $\ell_v$ and the point representing $v$ lies to the right of $\ell_v$ we have that $\lambda(u,v)=R$, otherwise we have that $\lambda(v,u)=R$. Next, traverse the cycle $c$ starting from $v$, proceeding in the direction opposite to $u$. Let $(w,z)$ be the first edge intersected by $\ell_v$ after $(u,v)$. If the point representing $w$ lies to the right of $z$, then we have $\lambda(w,z)=L$, otherwise $\lambda(z,w)=L$.


Second, we prove that if $c$ is complete it is rectilinear drawable. The proof is by induction on the size $p$ of $c$.
For the sake of simplicity, and since $\lambda(u,v)=\overline{\lambda(v,u)}$ we relabel all the edges of $c$ as follows. If the edge connecting $v_i$ and $v_{i+1}$ is directed from $v_i$ to $v_{i+1}$ we keep its label unchanged, else we assign to it $\overline{\lambda(v_{i+1},v_i)}$. After the relabeling the statement can be recast as follows: a shaped simple cycle $C=v_0, v_1, \dots, v_{p-1}$ with $p \geq 4$ is rectilinear drawable if and only if its edges contain at least once all labels in $\cal L$. 

In the base case, we have $p=4$. Since the shape $\lambda$ assigned to the edges of $c$ contains all four labels of  ${\cal L} = \{ L, R, D, U\}$, and since for each vertex $v_i$ we have $\lambda(v_{i-1},v_i) \neq \overline{\lambda(v_i,v_{i+1})}$, it follows that opposite directions cannot appear consecutively. Hence, the only two possible circular sequences of labels associated with the edges of $c$ are $R$-$U$-$L$-$D$ and $R$-$D$-$L$-$U$.
It follows that a rectilinear drawing of $c$ that respects $\lambda$  can be constructed by placing the four vertices at the corners of an axis-aligned rectangle, with each edge drawn as a horizontal or vertical segment, matching its assigned direction.

Without loss of generality, we can assume that there are no two consecutive edges of~$c$ with the same label. Therefore, $c$ alternates labels in $\{R,L\}$ with labels in $\{U,D\}$ and has even length. 
It follows that, for the inductive case, we have $p \geq 6$.
It is easy to see that the labeling $\Lambda = s_0, s_1, \dots, s_{p-1}$ of $c$ satisfies the hypotheses of \cref{le:reduction}. Therefore, it is always possible to find two labels $s_{i}$ and $s_{i+1}$ whose removal generates a subsequence $\Lambda' = s_0, s_1, \dots, s_{i-1}, s_{i+2}, \dots, s_{p-1}$ of length $p-2$, which also satisfies Properties~\textbf{\textit{a}}--\textbf{\textit{b}} of \cref{le:reduction}. Sequence $\Lambda'$ can be interpreted as a labeling of a cycle $C'$ of length $p-2$ and 
by induction, $C'$ admits a rectilinear drawing $\Gamma'$ shaped according to $\Lambda'$. We extend $\Gamma'$ to obtain a drawing $\Gamma$ of $c$ shaped according to $\Lambda$, by reintroducing the labels $s_{i}$ and $s_{i+1}$ in $\Lambda$ and by inserting the corresponding edges $(v_{i},v_{i+1})$ and $(v_{i+1},v_{i+2})$ in $\Gamma$ as follows. Denote by $(v_{i-1},v^*)$ (resp., by $(v^*,v_{i+3})$) the edge of $C'$ corresponding to $s_{i-1}$ (resp., $s_{i+2}$), where $v^*$ represents the vertices $v_i$, $v_{i+1}$ and $v_{i+2}$ that are merged together from the removal of the two edges. 
We describe the procedure for $s_{i-1} = R$ and $s_{i+2} = U$, the other cases being analogous up to a rotation or a flip. We have four cases for the subsequence $\langle s_i,s_{i+1}\rangle$: either it is $\langle U,R\rangle$, or $\langle U,L\rangle$, or $\langle D,R\rangle$, or $\langle D,L\rangle$ (see \cref{fig:insertion}). The drawing $\Gamma'$ of edges $(v_{i-1},v^*)$ and $(v^*,v_{i+3})$ is shown in \cref{fig:Gamma'}. To prevent both node-to-node and node-to-edge overlaps in $\Gamma$, we apply the following geometric transformations to $\Gamma'$.


Let $s_{i} = D$ and $s_{i+1} = L$. As shown in \cref{fig:insertion_D-L} :
\begin{enumerate}
\item each node $v_j$ such that $x(v_j) > x(v^*)$ is shifted one unit to the right;
\item each node $v_l$ such that $y(v_l)<y(v^*)$ is shifted one unit down.

\item Each horizontal edge $(v_k, v_{k+1})$ such that $x(v_k) \leq x(v^*)$ and $x(v_{k+1}) > x(v^*)$ is stretched by one unit 
\item Each vertical edge $(v_k, v_{k+1})$ such that $y(v_k) < y(v^*)$ and $ y(v_{k+1})\geq y(v^*)$ is stretched by one unit 
\item Each edge $(v_k, v_{k+1})$ such that $x(v_k) \geq x(v^*)$ and $x(v_{k+1}) \geq x(v^*)$ is shifted one unit to the right
\item Each edge $(v_k, v_{k+1})$ such that $y(v_k) \leq y(v^*)$ and $y(v_{k+1}) \leq y(v^*)$ is shifted one unit down 
\end{enumerate} \end{proof}

\thDrawable*
\begin{proof}
    The necessity follows from \cref{th:cycle}.
    
    In the remainder, we prove the sufficiency. Let $G$ be a shaped graph, and let $G_x$ and $G_y$ be its auxiliary directed graphs, defined as above. From the proof of \cite[Theorem 3]{DBLP:conf/gd/ManuchPPT10} we have that the presence of a cycle in $G_x$ or $G_y$ implies that $G$ is not rectilinear drawable while respecting its prescribed shape. We prove that if there is a cycle in $G_x$ or in $G_y$ then at least a simple cycle in $G$ is not complete.
    Assume, for a contradiction, that there is a cycle in $G_x$ or $G_y$ but all cycles of $G$ are complete. We show the contradiction for the cycle being in $G_x$, the other case is symmetric.
    Let $C_x = \mu_0,\dots,\mu_{p-1}$ be a cycle in $G_x$, possibly p=1.
    See \cref{fig:cycle-auxiliary}, where $p=2$.
    By the construction of $G_x$, for $i\in\{0,\dots,p-1\}$, there exist $p$ pairs $u_i,v_{i+1}$ of vertices of $V(G)$, such that $u_i \in \mu_i $ and $v_{i+1} \in \mu_{i+1}$, $(u_i,v_{i+1})\in E(G)$, and $\lambda(u_i,v_{i+1})=R$. Indices are taken $\text{mod~} p$.
    Again, by the construction of $G_x$, for each $v_i,u_i \in \mu_i$ there exists a path $\Pi_i$ in $G$ from $v_i$ to $u_i$ that contains edges all labeled $D$ or all labeled $U$. Consider the concatenation of the paths and edges in the following order: $\Pi_0,(u_0,v_1),\Pi_1,\dots,(u_{p-2},v_{p-1}),\Pi_{p-1},(u_{p-1},v_0)$. This concatenation forms a simple cycle $c$ in $G$. The cycle of \cref{fig:cycle-auxiliary} contains vertices $8$, $11$, $4$, $10$, $5$, $1$, $0$, $6$, $9$, and $7$.
    The constructed cycle $c$ is not complete, as it only contains edges labeled $R$, $U$, or $D$, and contains no edge with label $L$. This contradicts the assumption that all simple cycles in $G$ are complete, concluding the proof.
    The construction of $G_x$ and of $G_y$ and the acyclicity test can be done in linear time in the size of $V(G)$. Also, if the acyclicity test fails, a cycle $C_x$ can be found in linear time. Finally, given $C_x$ the cycle $c$ can be found in linear time.\end{proof}

\section{Definition of ${\cal F}_{G, \cal C}$}\label{sec:sat-formula-definition}

Formula ${\cal F}_{G, \cal C}$ has four boolean variables for each edge $(v,w)$ of $E(G)$. We call them $\ell_{v,w}$,~$r_{v,w}$,~$d_{v,w}$, and~$u_{v,w}$.
We have that:
\begin{enumerate}
    \item $\ell_{v,w}=true$ iff $\lambda(v,w)=L$,
    \item $r_{v,w}=true$ iff $\lambda(v,w)=R$,
    \item $d_{v,w}=true$ iff $\lambda(v,w)=D$, and
    \item $u_{v,w}=true$ iff $\lambda(v,w)=U$
\end{enumerate}


To simplify the description, we also introduce the variable $\ell_{w,v}$~(resp.~$r_{w,v}$), which is equal to~$r_{v,w}$~(resp.~$\ell_{v,w}$). Analogously, we introduce the variable $d_{w,v}$~(resp.~$u_{w,v}$), which is equal to $u_{v,w}$~(resp.~$d_{v,w}$).

Formula ${\cal F}_{G, \cal C}$ has three sets of clauses. 

The first set of clauses ensures that each edge $(v,w)$ is assigned exactly one label. Namely, for each edge $(v,w)$ we have the following clauses:
\begin{enumerate*}[label=\bf(\roman*)]
    \item\label{cla:at-least} $(\ell_{v,w} \vee r_{v,w} \vee d_{v,w} \vee u_{v,w})$,
    \item\label{cla:at-most-first} $(\neg \ell_{v,w} \vee \neg r_{v,w})$, 
    \item $(\neg d_{v,w} \vee \neg \ell_{v,w})$, 
    \item $(\neg d_{v,w} \vee \neg r_{v,w})$, 
    \item $(\neg u_{v,w} \vee \neg \ell_{v,w})$, 
    \item $(\neg u_{v,w} \vee \neg r_{v,w})$,
    \item\label{cla:at-most-last} $(\neg u_{v,w} \vee \neg d_{v,w})$.
\end{enumerate*}
Clause \ref{cla:at-least} is true if at least one label is assigned to $(v,w)$, clauses \ref{cla:at-most-first}--\ref{cla:at-most-last} are true if at most one label is assigned to $(v,w)$.

    


The second set of clauses guarantees that for each vertex $v$ with at least two neighbors, there are no two of such neighbors $w'$ and $w''$ of $v$ such that $\lambda(v,w')=\lambda(v,w'')$. Let $u_0,\dots,u_k$ be the neighbors of $v$, with $k=1,2,3$ according to the degree of $v$ minus one.
We distinguish two cases.
If $v$ has degree $4$, we have $4$ clauses, 
    $(\ell_{v,u_0} \vee \ell_{v,u_1} \vee \ell_{v,u_2} \vee \ell_{v,u_3})$,
    $(r_{v,u_0} \vee r_{v,u_1} \vee r_{v,u_2} \vee r_{v,u_3})$,
    $(d_{v,u_0} \vee d_{v,u_1} \vee d_{v,u_2} \vee d_{v,u_3})$, and
    $(u_{v,u_0} \vee u_{v,u_1} \vee u_{v,u_2} \vee u_{v,u_3})$, that are satisfied if for each label $l \in \cal L$ there exists an index $i$ ($i=0,\dots,3$) such that $\lambda(v,u_i)=l$.
Otherwise ($v$ has degree $2$ or $3$), for each pair of neighbors $u_i$ and $u_j$ of $v$, we have
    $(\neg \ell_{v,u_i} \vee \neg \ell_{v,u_j})$,
    $(\neg r_{v,u_i} \vee \neg r_{v,u_j})$,
    $(\neg d_{v,u_i} \vee \neg d_{v,u_j})$, and
    $(\neg u_{v,u_i} \vee \neg u_{v,u_j})$.

The third set of clauses guarantees that all  simple cycles in $\cal C$ are complete. For each simple cycle $v_0, v_1, \dots, v_{n-1}$ of ${\cal C}$ we impose that it contains all labels in ${\cal L}$ while traversing it. This is done with four clauses as follows:
    $(\ell_{v_0,v_1} \vee \ell_{v_1,v_2} \vee \dots \vee \ell_{v_{n-2},v_{n-1}} \vee \ell_{v_{n-1},v_0})$,
    $(r_{v_0,v_1} \vee r_{v_1,v_2} \vee \dots \vee r_{v_{n-2},v_{n-1}} \vee r_{v_{n-1},v_0})$,
    $(d_{v_0,v_1} \vee d_{v_1,v_2} \vee \dots \vee d_{v_{n-2},v_{n-1}} \vee d_{v_{n-1},v_0})$, and
    $(u_{v_0,v_1} \vee u_{v_1,v_2} \vee \dots \vee u_{v_{n-2},v_{n-1}} \vee u_{v_{n-1},v_0})$

\section{A Note on How Metrics are Measured}\label{sec:appendix-measures}

Let us start from the drawings of graphs with vertices of maximum degree $4$.

The $x$- and $y$-coordinates of the vertices and bends of the drawings produced by \DOMUS start from $0$ and have integer values. Also, let $x_M$ be the maximum used $x$-coordinate, we have that for each $i$ with $0 \leq i \leq x_M$, there is at least one vertex or bend whose $x$-coordinate is $i$. Analogously, let $y_M$ be the maximum used $y$-coordinate, we have that for each $j$ with $0 \leq j \leq y_M$, there is at least one vertex or bend whose $y$-coordinate is $j$. 

Hence, the area of a drawing is computed as $(x_M+1) \dot (y_M+1)$. Also, the length of an edge is computed as the sum of the lengths of the segments that represent that edge.

The $x$- and $y$-coordinates of the vertices and bends of the drawings produced by \OGDF, and saved in the corresponding .gml files, have the following features. Let us order the vertices and the bends of a drawing according to their $x$-coordinates. We have that two consecutive vertices/bends in this order with different $x$-coordinates either have $x$-coordinates that differ for about $30$ units or that differ for about $1$-$5$ units. The first, are intended to stay on different columns and the second, with small variations, are intended to stay on the same column. It follows that, it would be unfair to apply to \OGDF the same formulas we apply for \DOMUS. Hence, we normalize the $x$-coordinates given by \OGDF moving the $x$-coordinate of the first vertex to $0$, transforming the $30$ unit gaps into unit gaps, and giving the same $x$-coordinates to the vertices that have about the same \OGDF $x$-coordinates. We apply the same algorithm to the $y$-coordinates and then use the formulas described above.
As an example, consider the drawing produced by \OGDF in \cref{fig:snap-OGDF-to-grid}. We have that: (1) vertices $14$, $18$, and $19$ have all the same $x$-coordinate; (2) the same is true for vertices $1$, $9$, $6$, and $11$; (3) the difference of $x$-coordinate between vertex $19$ and vertex $11$ is $30$; (4) the two left bends have the same $x$-coordinate and the difference with vertex $19$ is $25$. Hence, we assign $x$-coordinate $0$ to the two bends, $x$-coordinate $1$ to $14$, $18$, and $19$, $x$-coordinate $2$ to $1$, $9$, $6$, and $11$. Notice that vertices $5$ and $15$ are slightly misaligned ($4$ units apart). We assign to them the same $x$-coordinate, which is $3$.

\begin{figure}[tb!]
\centering
    \includegraphics[width=0.3\textwidth]{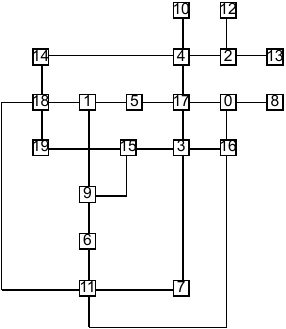}
    \caption{A drawing produced by \OGDF.}
    \label{fig:snap-OGDF-to-grid}
\end{figure}

Both in \DOMUS and in \OGDF drawings of graphs with vertices with degree greater than $4$ we have very near segments representing edges exiting a vertex in the same direction. In computing the above metrics we consider all those segments aligned with the coordinates of the staring vertex.

\section{Other Models for Orthogonal Drawings}\label{app:hola}

As discussed in \cref{sec:experiments-4}, we 
couldn't compare the drawings computed with \DOMUS with those computed with recent tools that have been evaluated with user studies, as, for example, Hola~\cite{Kieffer-2016}.
In fact, these tools adopt a model for orthogonal drawings that differs from the traditional one. For example: vertices with degree less or equal than 4 are allowed to have edges incident on the same side of the vertex box (see \cref{fig:hola} for some examples). This freedom results in more compact drawings but may have the side effect of causing some edges to run very close to each other (see, for example, the L-shaped edges on the top-left side of \cref{fig:hola-1}). Further, although vertices seem to be regularly distributed on the plane, it is apparent that they are not constrained to be on a grid, and sometimes intermediate coordinates seem to be freely used. The same freedom seems to be used for bends, that sometimes are very close to the vertices boxes (see, for example, the edge at the top of \cref{fig:hola-2}). We are aware that these relaxations of the stricter orthogonal model are meant to obtain more pleasant and symmetric drawings. However, we believe that the drawings obtained within the two different models are not comparable with traditional metrics such as number of bends, edge length, area, and the like.   

\begin{figure}[tb!]        
    
    \begin{subfigure}[b]{0.38\textwidth}
    \centering
        \includegraphics[page=1,width=\textwidth]{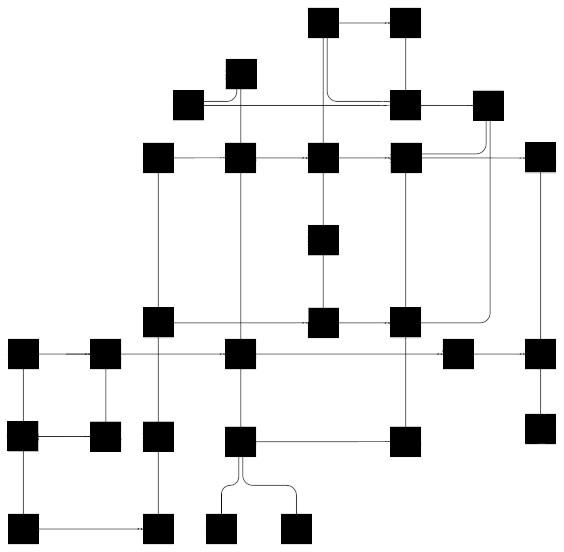}
        \subcaption{}
        \label{fig:hola-1}
    \end{subfigure}
            \hfill
    \begin{subfigure}[b]{0.45\textwidth}
    \centering
        \includegraphics[page=2,width=\textwidth]{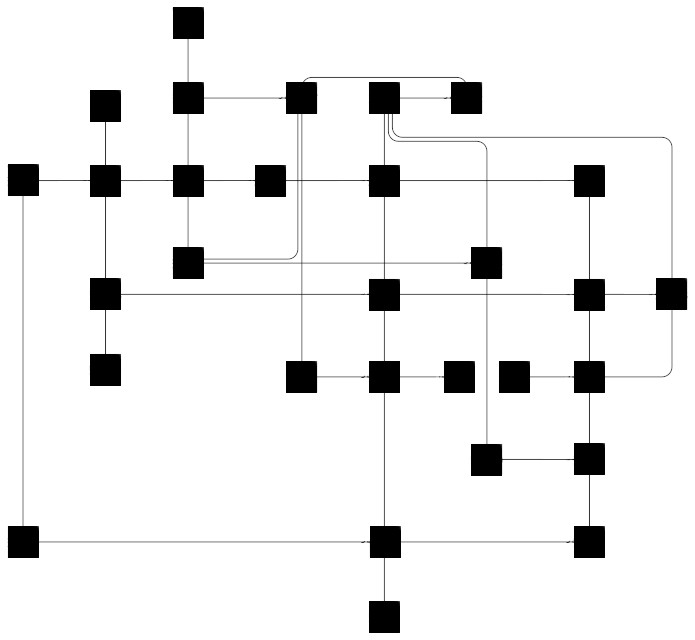}
        \subcaption{}
        \label{fig:hola-2}
    \end{subfigure}
            \hfill
    
    \caption{Two orthogonal drawings produced by the Adaptagrams \texttt{libdialect} library, implementing the Hola model~\cite{Kieffer-2016}.}
    \label{fig:hola}
\end{figure}

\section{Extra Diagrams from \cref{sec:experiments-4} (\nameref{sec:experiments-4})}\label{app:missing-figures-sec-5}

\cref{fig:scatter_comparisons-part-2} shows Total Edge Length, Max Edge Length, and Edge Length deviation are independent of the density. The area is instead dependent on small values of $n$.

\begin{figure}[h]
    \begin{subfigure}[b]{0.48\textwidth}
    \centering
        \includegraphics[width=\textwidth]{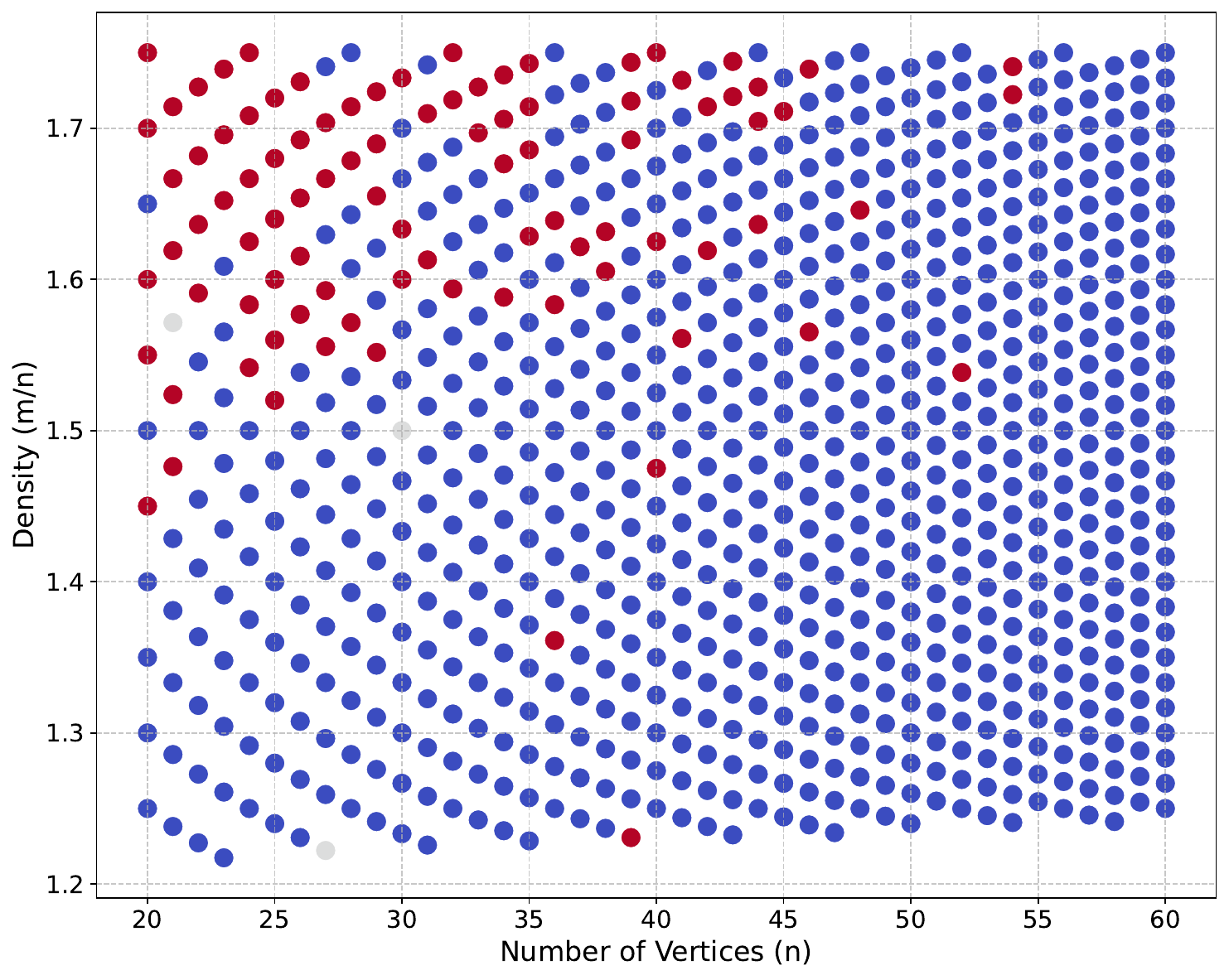}
        \subcaption{Area}
        \label{fig:scatter_area_comparison}
    \end{subfigure}
            \hfill
    \begin{subfigure}[b]{0.48\textwidth}
    \centering
        \includegraphics[width=\textwidth]{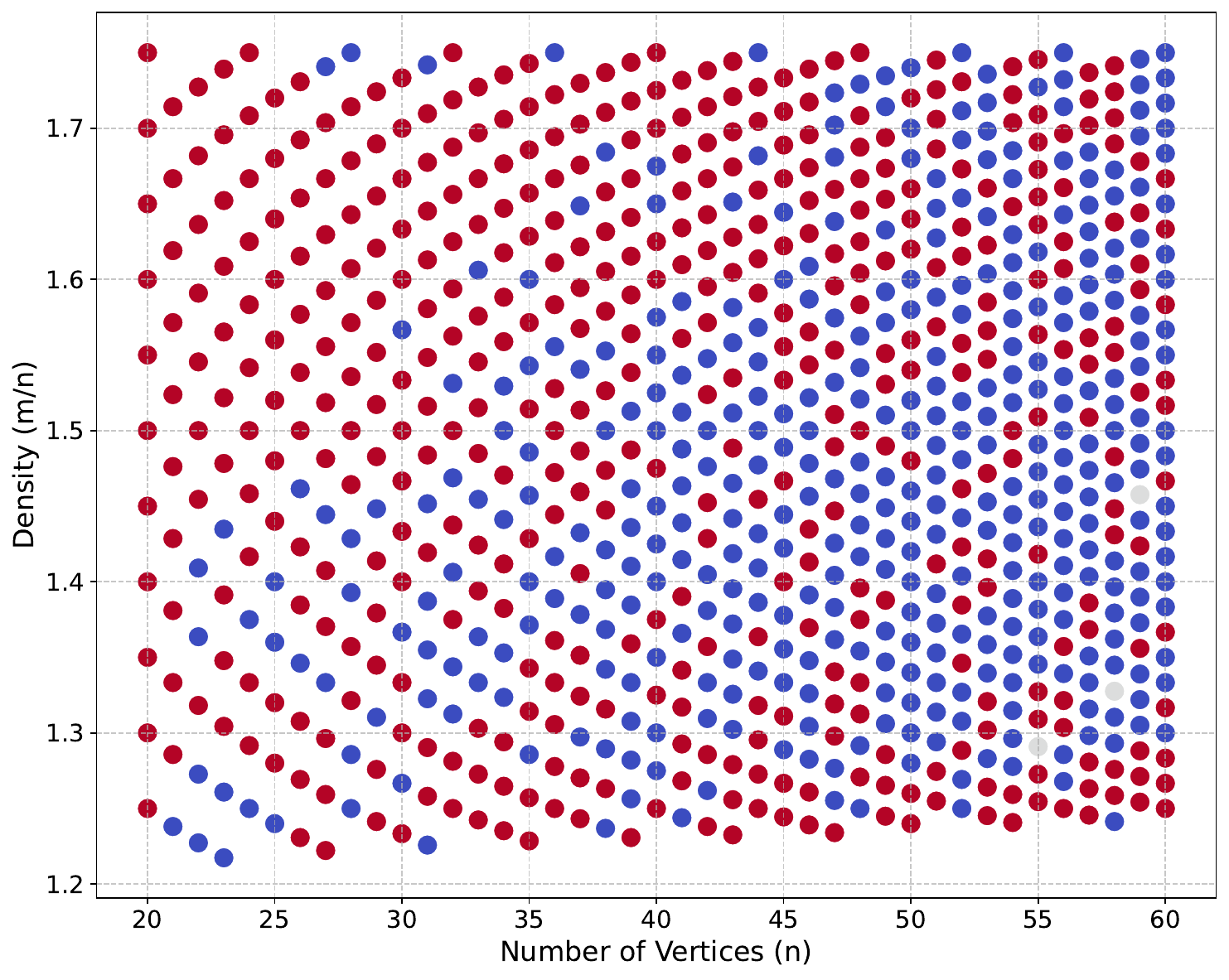}
        \subcaption{Total Edge Length}
        \label{fig:scatter_edge_length_comparison}
    \end{subfigure}

    \par\smallskip
            
    \begin{subfigure}[b]{0.48\textwidth}
    \centering
        \includegraphics[width=\textwidth]{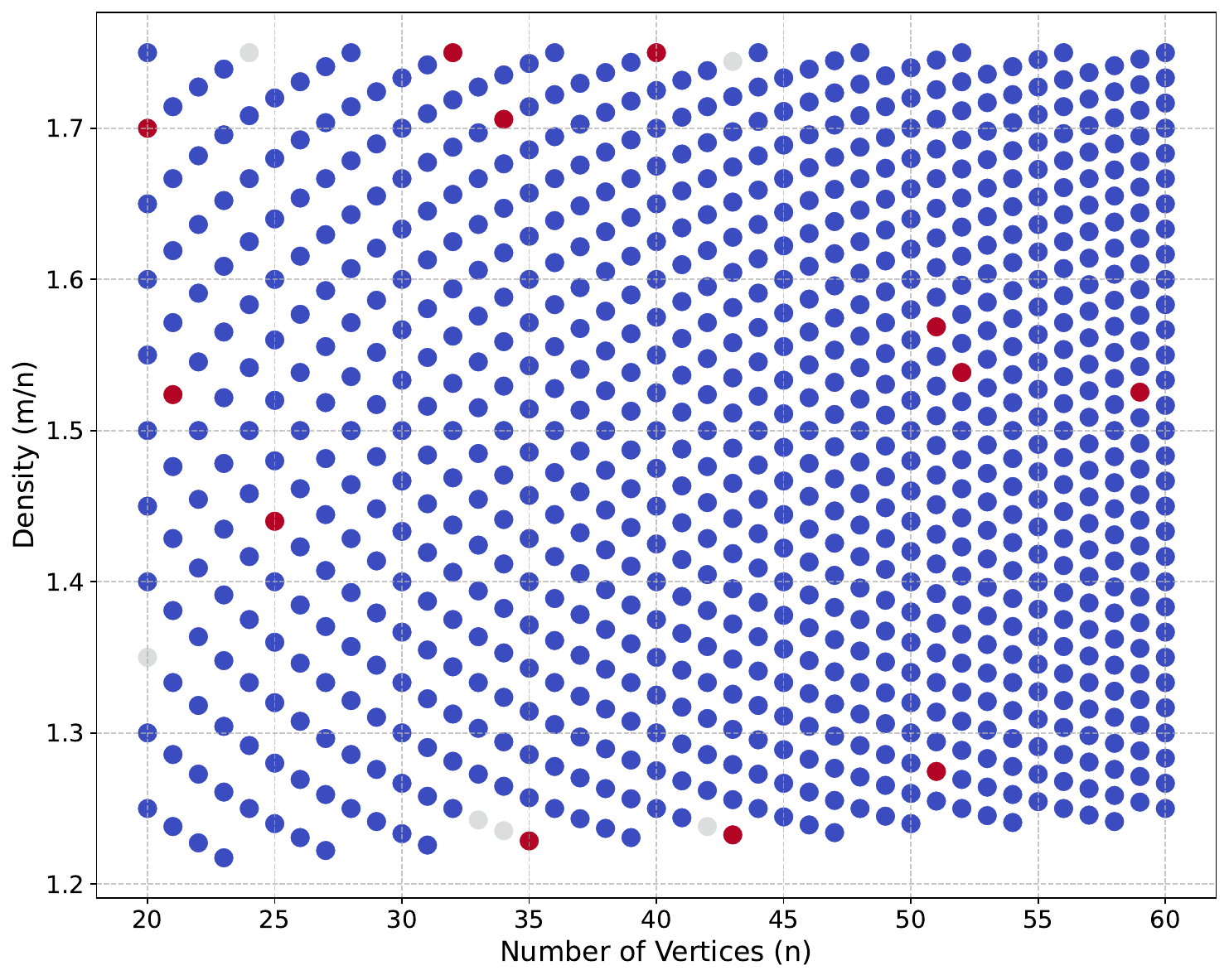}
        \subcaption{Max Edge Length}
        \label{fig:scatter_max_edge_length_comparison}
    \end{subfigure}
            \hfill
    \begin{subfigure}[b]{0.48\textwidth}
    \centering
        \includegraphics[width=\textwidth]{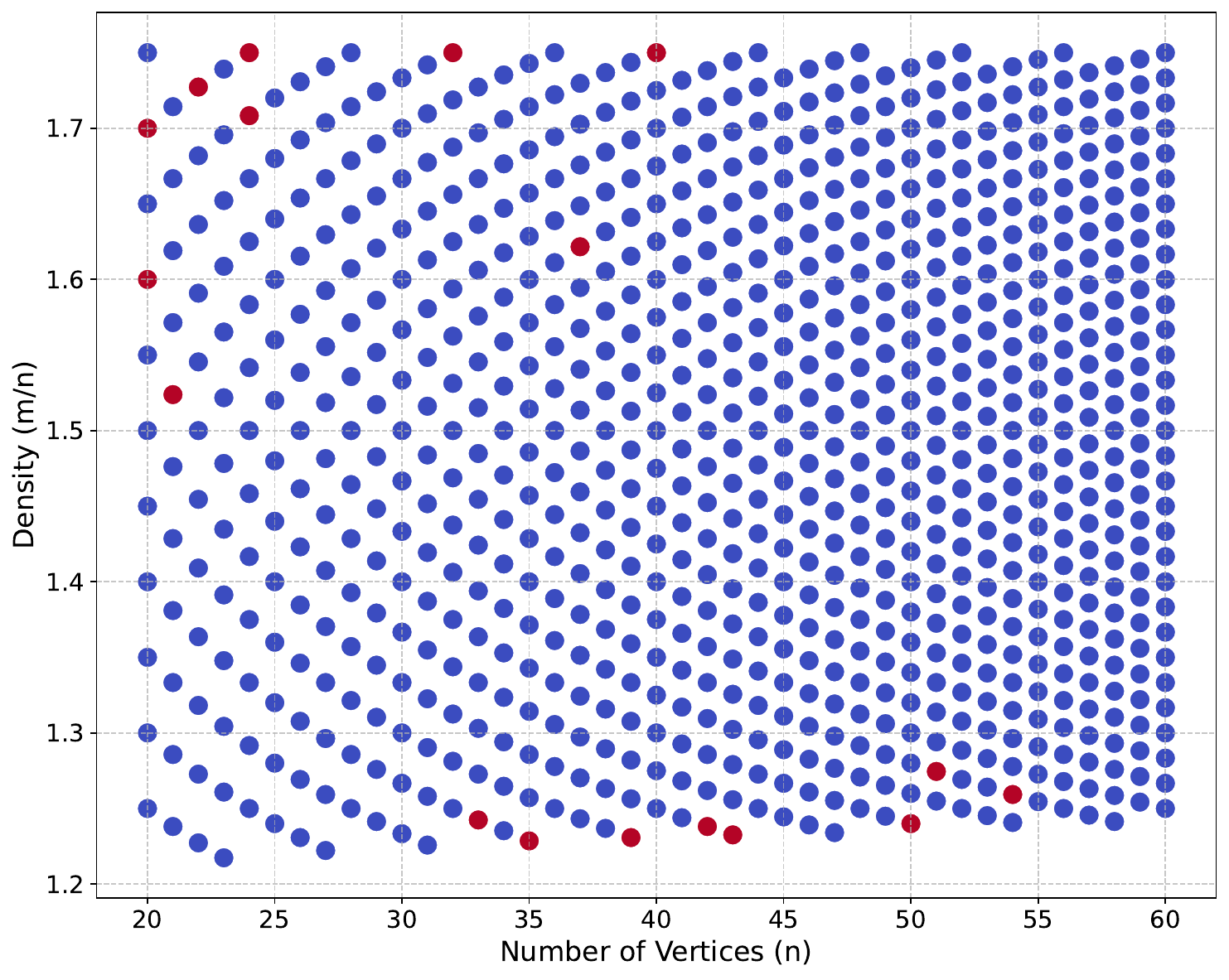}
        \subcaption{Edge Length Deviation}
        \label{fig:scatter_stddev_edge_length_comparison}
    \end{subfigure}
    \caption{The effect of density: Area and Edge Length. Bends and Crossings. Blue: \DOMUS is better. Red: \OGDF is better. Grey: parity.}
    \label{fig:scatter_comparisons-part-2}
\end{figure}


\section{Small Sets of Cycles for Computing a Rectilinear Drawable Shape}\label{app:small-set-of-cycles}


Although the problem of computing a rectilinear drawable shape is NP-complete, one may expect that the computation of a rectilinear drawable shape could neglect some cycle, whose completeness is guaranteed by the remaining cycles. Ideally, one would like to have to cope with a polynomial number of cycles. 
The following lemma shows that, if this restricted set of cycle exists, it is not its size to be important, but the kind of cycles that it contains.
Indeed, when computing a shape for a 4-graph $G$, the fact that a large number of simple cycles of $G$ are complete does not guarantee that all the simple cycles of $G$ are complete, to the point that neglecting a single cycle of the exponentially many of~$G$ could produce a shaped graph that is not rectilinear drawable. Conversely, once a single cycle is satisfied, a rectilinear drawable shape may become trivial to compute.   

\begin{lemma}\label{le:one-missing}
    There exists an infinite family of 4-graphs $\mathscr{G}_i$, with $i = 1, 2, 3, \dots$, such that $\mathscr{G}_i$ has $O(i)$ vertices and $O(i)$ edges and admits a shape that makes all its $O(2^i)$ simple cycles complete with the exception of a single one. 
\end{lemma}
\begin{proof}
    Graph $\mathscr{G}_i$ is composed by a cycle $c$ with $10 + 2i$ vertices $C = v_1, v_2, \dots, v_5,$ $y_1, y_2, \dots, y_i$, $u_5, u_4, \dots, u_1$, $x_i, x_{i-1}, \dots, x_1$ and by $5+i$ chordal paths composed of $5$ edges and denoted $\pi_{v_2u_2}, \pi_{v_3u_3}, \pi_{v_4u_4}$, $\pi_{v_1u_5}, \pi_{u_1u_5}$, and $\pi_{x_jy_j}$, for $j=1,2, \dots, i$, where $\pi_{ab}$ is the path that connects the pair of vertices $\{a,b\}$ of $c$ (refer to \cref{fig:counter-example-1}.
    Observe that $\mathscr{G}_i$ has $30+6i$ vertices and $35+7i$ edges. In order to show that $\mathscr{G}_i$ has $O(2^i)$ cycles, we construct a different cycle $C_B$ for each possible assignment of $i$ Boolean values $B = \langle b_1, b_2, \dots, b_i \rangle$. We denote by $v_{\ell}$ the last vertex added to $C_B$. Start from the edge $(v_1, x_1)$ and $v_{\ell} = x_1$. If $b_1 = \texttt{true}$ add to $C_B$ the chord $\pi_{x_1y_1}$ and the edge $(y_1,y_2)$, otherwise, add the edge $(x_1, x_2)$. Observe that $v_{\ell}$ is either $x_2$ or $y_2$. Now, if $b_2 = \texttt{true}$, add the chord $\pi_{x_2y_2}$ and the edge that leads to $x_3$ or $y_3$, otherwise move directly to $x_3$ or to $y_3$. Observe that $v_{\ell}$ is either $x_3$ or $y_3$. Carry on constructing $C_B$ based on the values $b_j$ until either $x_i$ or $y_i$ is reached. Now add to $C_B$ the shortest path from $v_{\ell}$ to $u_3$ in $\mathscr{G}_i$, add path $\pi_{u_3v_3}$, and add the shortest path from $u_3$ to $v_1$. Since each sequence $C_B = \langle b_1, b_2, \dots, b_i \rangle$ of $i$ Boolean values produces a different cycle $C_B$, the number of cycles of $\mathscr{G}_i$ is $O(2^i)$.  
    In order to prove the statement we exhibit a shape that makes all cycles complete with the exception of $c$ (refer to \cref{fig:counter-example-2}). 
    \begin{enumerate}
    \item Label the path $x_1, v_1, v_2, \dots, v_5, y_1$ with labels $R,D,L,L,D,R$
    \item Label the path $y_1, y_2, \dots, y_i$ with labels $R, R, R, \dots, R$
    \item Label the path $y_i, u_5, u_4, \dots, u_1, x_i$ with labels $R, D, L, L, D, R$
    \item Label the path $x_i, x_{i-1}, x_{i-2}, \dots, x_1$ with labels $R, R, R, \dots, R$
    \item Label $\pi_{v_1v_5}$ with labels $R,U,L,D,R$
    \item Label $\pi_{u_1u_5}$ with labels $D, R, U, L, D$
    \item Label $\pi_{v_2u_2}$ with labels $D, R, U, L, D$
    \item Label $\pi_{v_3u_3}$ with labels $U, R, D, L, U$
    \item Label $\pi_{v_4u_4}$ with labels $U, R, D, L, U$
    \item Label each $\pi_{x_jy_j}$, with $j=1,2,3, \dots, i$ with labels $U,R,D,L,U$  
    \end{enumerate}

It can be observed (refer to \cref{fig:counter-example}) that the described labeling is a shape (no two edges leave the same vertex with the same label; no to adjacent edges have opposite labels).
Also, cycle $c$ is missing the label $U$ (items 1, 2, 3, and 4 of the list above) and, hence, it is not complete. 
Any other cycle of $\mathscr{G}_i$ uses at least a chordal path, and each chordal path contains all the four labels (items 5 to 10). Therefore, every other simple cycle of $\mathscr{G}_i$ is complete. 
\end{proof}

Observe that the graph $\mathscr{G}_i$ described in the proof of \cref{le:one-missing} is planar. Also, its maximum degree is actually $3$. 
Further, observe that if the edges of the cycle $c$ were labeled in such a way to make $c$ complete, it would be trivial to assign labels to the other chordal paths so to make all cycles complete, as each chordal path has five edges and could easily be labeled with four different labels. 

\begin{figure}[tb!]        
    
    \begin{subfigure}[b]{0.48\textwidth}
    \centering
        \includegraphics[page=1,width=\textwidth]{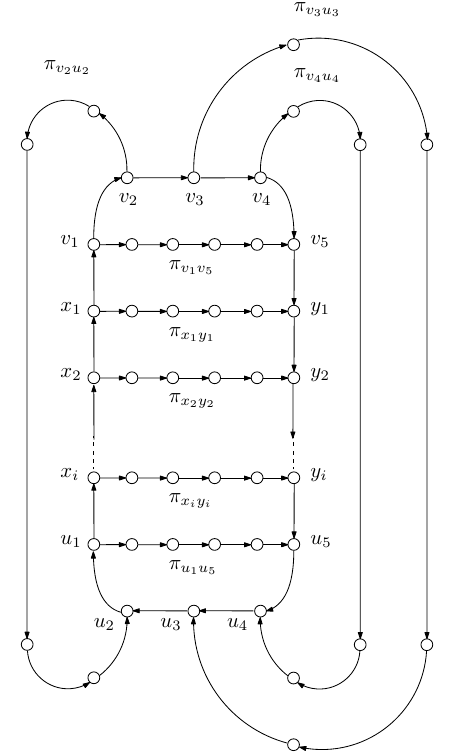}
        \subcaption{}
        \label{fig:counter-example-1}
    \end{subfigure}
            \hfill
    \begin{subfigure}[b]{0.48\textwidth}
    \centering
        \includegraphics[page=2,width=\textwidth]{img-counter-example.pdf}
        \subcaption{}
        \label{fig:counter-example-2}
    \end{subfigure}
            \hfill
    
    \caption{A figure for \cref{le:one-missing}.}
    \label{fig:counter-example}
\end{figure}

\end{document}